\documentclass[journal,10pt]{IEEEtran}
\usepackage{cite}
\usepackage{amsmath,amssymb,amsfonts}
\usepackage{algorithmic}
\usepackage{graphicx}
\usepackage{amsthm}
\usepackage{epsfig}
\usepackage{color}
\usepackage{subcaption}
\usepackage{lettrine}
\usepackage{float}
\usepackage{lipsum}
\usepackage{gensymb}
\usepackage{bm}
\usepackage[binary-units=true]{siunitx}
\usepackage{mathtools}
\usepackage{cuted}
\usepackage{graphics}
\usepackage{bbm}
\usepackage{suffix}
\usepackage[T1]{fontenc}
\usepackage[colorlinks]{hyperref} 
\usepackage[nameinlink,noabbrev]{cleveref}
\usepackage{textcomp}
\usepackage{xcolor}
\usepackage[def-file=diffcoeff-doc]{diffcoeff}
\usepackage{stackengine}
\usepackage[ruled,norelsize]{algorithm2e}

\usepackage{physics}
\usepackage[colorinlistoftodos]{todonotes}

\usepackage{textcomp}
\def\BibTeX{{\rm B\kern-.05em{\sc i\kern-.025em b}\kern-.08em
    T\kern-.1667em\lower.7ex\hbox{E}\kern-.125emX}}

\newtheorem{lemma}{Lemma}

\newtheorem{remark}{Remark}

\newcommand{\K}{\ensuremath{\mathcal K}}
\newcommand{\N}{\ensuremath{\mathcal N}}
\renewcommand{\L}{\ensuremath{\mathcal L}}

\newcommand{\e}{\ensuremath{\mathbf{e}}}
\newcommand{\W}{\ensuremath{\mathcal W}}
\newcommand{\Pow}{\ensuremath{\mathbf{P}}}
\renewcommand{\P}{\ensuremath{\mathrm{P}}}
\renewcommand{\r}{\ensuremath{\mathbf{r}}}

\newcommand{\C}{\ensuremath{\mathrm C}}
\newcommand{\Q}{\ensuremath{\mathbf{Q}}}
\renewcommand{\u}{\ensuremath{\mathbf{u}}}
\newcommand{\T}{\ensuremath{\mathbf{T}}}

\newcommand{\D}{\ensuremath{\mathcal D}}
\renewcommand{\H}{\ensuremath{\mathcal H}}
\newcommand{\I}{\ensuremath{\mathbf{I}}}

\newcommand{\q}{\ensuremath{\mathbf q}}
\renewcommand{\v}{\ensuremath{\mathbf v}}
\newcommand{\E}{\ensuremath{\mathbb E}}
\DeclareMathOperator{\mydiag}{\mathbf{diag}}  
\DeclareMathOperator{\mytrace}{\mathbf{tr}}  
\DeclareMathOperator{\myrank}{\mathbf{rank}}  
\renewcommand{\O}{\ensuremath{\mathcal O}}

\def \treq {\stackrel{\tiny \Delta}{=}}

\makeatletter
\def\@seccntformat#1{\@ifundefined{#1@cntformat}%
	{\csname the#1\endcsname\quad}
	{\csname #1@cntformat\endcsname}
	}
\newcommand{\removelatexerror}{\let\@latex@error\@gobble}
\makeatother

\usepackage{chngcntr}
\counterwithin*{subsection}{section}
\renewcommand{\thesubsection}{\Alph{subsection}}

\usepackage{scalerel}
\usepackage{tikz}
\usetikzlibrary{svg.path}

\definecolor{orcidlogocol}{HTML}{A6CE39}
\tikzset{
  orcidlogo/.pic={
    \fill[orcidlogocol] svg{M256,128c0,70.7-57.3,128-128,128C57.3,256,0,198.7,0,128C0,57.3,57.3,0,128,0C198.7,0,256,57.3,256,128z};
    \fill[white] svg{M86.3,186.2H70.9V79.1h15.4v48.4V186.2z}
                 svg{M108.9,79.1h41.6c39.6,0,57,28.3,57,53.6c0,27.5-21.5,53.6-56.8,53.6h-41.8V79.1z M124.3,172.4h24.5c34.9,0,42.9-26.5,42.9-39.7c0-21.5-13.7-39.7-43.7-39.7h-23.7V172.4z}
                 svg{M88.7,56.8c0,5.5-4.5,10.1-10.1,10.1c-5.6,0-10.1-4.6-10.1-10.1c0-5.6,4.5-10.1,10.1-10.1C84.2,46.7,88.7,51.3,88.7,56.8z};
  }
}

\newcommand\orcidicon[1]{\href{https://orcid.org/#1}{\mbox{\scalerel*{
\begin{tikzpicture}[yscale=-1,transform shape]
\pic{orcidlogo};
\end{tikzpicture}
}{|}}}}

\usepackage{hyperref} 

\IEEEoverridecommandlockouts
\begin{document}

\title{ Aerial Intelligent Reflecting Surface Enabled Terahertz Covert Communications in Beyond-5G Internet of Things}

\author{Authors}
\author{\vspace{-1mm}Milad Tatar Mamaghani\textsuperscript{\orcidicon{0000-0002-3953-7230}}\,,~\IEEEmembership{Graduate Student Member,~IEEE}, Yi Hong\textsuperscript{\orcidicon{0000-0002-1284-891X}}\,,~\IEEEmembership{Senior Member,~IEEE}%
 \thanks{
 \par Milad   Tatar   Mamaghani   and   Yi   Hong   are   with   the   Department of   Electrical   and   Computer   Systems   Engineering,   Faculty   of   Engineering,   Monash   University,   Melbourne,   VIC   3800,   Australia   (corresponding author e-mail:   milad.tatarmamaghani@monash.edu). 
 This work is supported by the Australian Research Council (ARC) through the ARC discovery project DP210100412.
}
}
\markboth{}{}
\maketitle
\begin{abstract}
Unmanned aerial vehicles (UAVs) are envisioned to be extensively employed for assisting wireless communications in the Internet of Things (IoT). 
On the other hand, terahertz (THz) enabled intelligent reflecting surface (IRS) is expected to be one of the core enabling technologies for forthcoming beyond-5G wireless communications that promise a broad range of data-demand applications. In this paper, we propose a UAV-mounted IRS (UIRS) communication system over THz bands for confidential data dissemination from an access point (AP) towards multiple ground user equipments (UEs) in IoT networks. Specifically, the AP intends to send data to the scheduled UE, while unscheduled UEs may behave as potential adversaries. \textcolor{black}{To protect information messages from the privacy preservation perspective}, we aim to devise an energy-efficient multi-UAV covert communication scheme, where the UIRS is for reliable data transmissions, and an extra UAV is utilized as an aerial cooperative jammer, opportunistically generating artificial noise (AN) to degrade unscheduled UEs detection, leading to communication covertness improvement.
This poses a novel max-min optimization problem in terms of minimum average energy efficiency (mAEE), aiming to improve covert throughput and reduce UAVs' propulsion energy consumption, subject to satisfying some practical constraints such as the covertness requirements for which we obtain analytical expressions. Since the optimization problem is non-convex, we tackle it via the block successive convex approximation (BSCA) approach to iteratively solve a sequence of approximated convex sub-problems, designing the binary user scheduling, AP's power allocation, maximum AN jamming power, IRS  beamforming,  and both UAVs' trajectory and velocity planning. Finally, we present a low-complex overall algorithm for system performance enhancement with complexity and convergence analysis. Numerical results are provided to verify the analysis and demonstrate significant outperformance of our design over other existing benchmark schemes concerning the mAEE performance.
\end{abstract}


\section{Introduction}
\lettrine[lines=2]{T}{he} 
wireless communication and networking architectures have been witnessed revolutionary progress over the past few years. Indeed, yesteryear's smartphone-centered networks have gradually enlarged to harmoniously integrate a heterogeneous combination of massive wireless-enabled device equipment ranging from smartphones, connected intelligent vehicles, and wearables, aiming at eventually realizing the truly connected \textit{Internet of Things (IoT)} systems in the form of \textit{Internet of Everything} \cite{Akyildiz2020}. 

This unprecedented proliferation of IoT devices will unquestionably drive exponential growth of wireless network traffic inasmuch as the next generation wireless systems need to offer not only larger system capacity with ultra-reliable and low-latency communication but more flexible networking capabilities to adaptively cater the requirements of the IoT's dynamics \cite{Xu2021b}. In particular, reliable data dissemination is one of the core pillars of an IoT-connected society, wherein the ground wireless access point (AP) typically needs to efficiently transmit information messages to IoT node(s) via multi-hop routing \cite{TatarMamaghani2021a}. Such static networking architecture may admit various shortcomings \cite{Li2020}. First, from the perspective of capital and operational expenditures, deploying and operating traditional terrestrial infrastructures in some areas such as mountainous terrain and marine regions can be costly and inefficient.
Additionally, the higher the number of intermediate nodes in multi-hop routing architectures, the more the network delay; thereby, they might be inappropriate for delay-sensitive IoT applications. Balancing the energy consumption of IoT devices is also imperative for prolonging the lifetime of IoT systems, particularly for IoT nodes having various distances to the AP. Indeed, the closer the IoT node to the AP, the quicker the energy depletion it may face due to heavier load bringing on an energy hole in the overall system. Last but not least, significant computation, synchronization, and control signaling are required for safeguarding such energy-hungry multi-hop IoT systems as they may be vulnerable to various adversary attacks due to the openness of wireless environments. To cope with the aforementioned challenges of IoT development, some promising solutions have recently been proposed in \cite{Jiang2021a, Wu2021}.

\subsection{UAV-IoT Communications}
Unmanned aerial vehicles (UAVs) have recently been identified as a promising technology for a myriad of civilian applications, so much so that the global market for the commercial UAV industry has been visioned to skyrocket some USD 55 billion by 2027 \cite{globalmarket_uav}. With on-demand swift deployment, low-cost operation and maintenance, flexible and controllable maneuverability, UAVs can be widely utilized in a broad range of scenarios such as goods shipment\footnote{Note that in the wake of the recent COVID-19 pandemic, some countries have practically deployed UAVs for lab sample and medical supplies delivery purposes to reduce transportation times and exposure rate \cite{covid19_uav2021}.}, real-time road traffic monitoring, precision agriculture,  remote sensing,  communication relaying, and wireless coverage \cite{Shakhatreh2019}.
Especially in the UAV-aided wireless communication paradigm, thanks to their flexible mobility, UAVs can establish strong line-of-sight (LoS) air-ground (AG) links towards the ground IoT devices, offering excellent wireless coverage and reduced energy consumption for such resource-constraint networks, and thus, overcoming the drawbacks of traditional fixed infrastructures.  To this end, proper path planning/deployment and resource management for the UAV-IoT networks are of significant importance to the extent that a great deal of research has been devoted to the design of such systems (see \cite{Xu2021b, Wan2020, mamaghani2020_tvt, Ding2020, TatarMamaghani2021, Lei2020} and references therein).  However, the majority of previous research efforts have considered utilizing the microwave spectrum bands of sub-6 GHz for UAV communications,  which has already been heavily occupied by traditional wireless systems leading to the spectrum crunch crisis \cite{chen2019survey}. Therefore, UAV-communication system designs on other frequency bands, such as the promising terahertz (THz) bands for the beyond-5G (B5G) UAV-IoT networks, are in demand.

\subsection{THz and IRS Technologies}
THz communication is celebrated to be one of the emerging technologies for B5G  wireless communications thanks to the abundance of unexplored available spectrum ($0.1-10\mathrm{THz}$) and the potentiality of fulfilling remarkable wireless capacity enhancement \cite{chen2019survey, Sarieddeen2020, Chaccour2021, Tekbiyik2020}. Recently, THz transmissions have been investigated for UAV communication applications  \cite{Xu2020,mamaghani2021terahertz}. \textcolor{black}{THz signals can pave the way for sharper directionality and may guarantee improved secure and reliable transmissions compared to traditional low radio frequency (RF) counterparts. However, such benefits come at the cost of some channel peculiarities, such as  highly nature of frequency selectivity, substantial path loss arising from distance-related attenuation as well as the water-vapor absorption phenomenon, which  generally  dependents  on  operational frequency,  distance,  altitude,  and relative  air  composition \cite{chen2019survey}}.

To compensate for their higher propagation attenuation and achieve sustainable capacity enhancement, utilizing \textit{intelligent reflecting surface (IRS)} \cite{Basharat2021, Wu2020a, Bjornson2020} has emerged as a possible solution for removing the barriers of relatively unreliable and costly conventional THz transmissions. Indeed, an IRS is a thin planar meta-surface composed of a large number of reconfigurable scatterers, each of which can independently collect the impinging RF signal, adjust its electromagnetic (EM) properties (e.g., the amplitude and phase shift) in real time under the control of a smart IRS controller, which can be implemented via a field-programmable gate array (FPGA), and then reflect it so as to obtain the desired realization. Plus, IRS can improve the spectral efficiency compared to the traditional half-duplex relaying \cite{TatarMamaghani2018, TatarMamaghani2019a, TatarMamaghani2021} and also does not incur additional cost for sophisticated self-interference cancellation algorithms conventionally utilized for full-duplex relaying \cite{Mamaghani2021}. An IRS passively beamforms signals without the need for any RF transceiver chains and accordingly offers the appealing feature of free of noise-corruption full-duplex relaying. Thus, several works have recently explored the benefits of deploying IRS from different perspectives. For example,  Wu and Zhang studied the joint passive and active beamforming design in an IRS-assisted MIMO system \cite{Wu2019a}.  Pan \textit{et al.} investigated the sum-rate maximization problem for an IRS-aided THz communication system including multiple users each demanding for a different rate, via jointly designing the IRS location and phase shift, sub-band allocation, and power control \cite{Pan2020}. Deshpande \textit{et al.} considered an energy-efficient design for terrestrial IRS-empowered UAV communications via both joint trajectory, transmission power, and the phase shift optimization of an IRS with a fixed location in \cite{Deshpande2021}. 
\textcolor{black}{The aforementioned works primarily focused on exploiting fixed IRS deployed on indoor walls or facades of buildings, which, in turn, poses fundamental limitations such as finding an appropriate installation place, $180\degree$ half-space reflection capability, and significant signal attenuation due to several reflections, particularly in complex urban environments \cite{Lu2020}. Nevertheless, being low profile and lightweight, an IRS can be integrated with aerial platforms such as UAVs to enable intelligent reflection from the sky \cite{long2020reflections}. Such an aerial IRS system can potentially offer $360\degree$ panoramic full-angle reflection, more flexible three-dimensional (3D) network design, and last but not least, stronger channels compared to traditional IRS deployments.}

\subsection{Secure and Covert Communications}
UAV communications are appealing in terms of capacity and coverage improvement thanks to the possibility of highly LoS air-ground (AG) links. Nonetheless, the open nature of such links exposes the security of UAV-aided wireless communications, which is of pivotal importance, at significant risk. This has recently gravitated the research community to incorporate the security paradigm for designing UAV communication systems by mainly utilizing the physical layer techniques \cite{Mamaghani2017, Mamaghani2019} to avoid extra signaling and overheads incurred by conventional upper-layer cryptography. The \textit{information-theoretic secrecy (ITS)} for safeguarding UAV communications has been extensively studied by joint design of trajectory and resource allocation in the literature \cite{Li2020, TatarMamaghani2021, Hong2019, Sun2019e, Zhong2019, Tang2019, Hongliang2018Sec}. However, preventing the content of information message from being deciphered by an eavesdropper for which the ITS aims might be inadequate when the privacy protection matters. As such, in some scenarios, the \textit{existence} of legitimate transmissions needs to be sheltered from a possible vigilant adversary (or the so-called warden), and communicating terminals may desire to transfer messages covertly, since the expose of legitimate transmissions might plausibly attract the warden's attention for launching possible hostile attacks \cite{Soltani2018Covert, Lu2020a, Forouzesh2021}.

\textcolor{black}{
Covert communications, a.k.a. low probability of detection (LPD) communications, have recently emerged  to address the ever-increasing desire for strong security and privacy in 5G-and-beyond wireless networks and IoT by hiding wireless transmissions \cite{Yan2019b}, and drawn significant interest amongst researchers who have established fundamental limits of the LPD communications by presenting the square law limit, which states that $\mathcal{O}(\sqrt{n})$ bits per $n$ channel uses can be reliably and covertly conveyed over the noisy channel between an intended source-destination pair without being detected by an adversary \cite{Bloch2016, Bash2013}. 
He \textit{et al.} explored covert communications considering distribution of  noise uncertainty in a statistical sense \cite{He2017a}. Adopting the technique of channel inversion power control via utilizing a full-duplex receiver for covert communications has been examined in \cite{Hu2019}.}
\textcolor{black}{ The authors in \cite{Kong2021} investigated a static IRS-assisted covert communication system and optimized the achievable rate for covert transmission. The problem of delay-sensitive covert communications over noisy channels with a finite block length has been explored in \cite{Yan2019}.}

\textcolor{black}{
Some recent works have viewed covert communications from the UAV perspective. For example,  performance analysis of terrestrial covert communication with the help of a UAV-aided artificial noise (AN) generation under different fading scenarios has been evaluated in \cite{Liang}. The authors in \cite{Chen2021} studied a four-node UAV-relaying covert communication scheme with finite blocklength to maximize the effective transmission bits between a legitimate source-destination pair against a flying warden.
In the presence of randomly distributed wardens, a covert communication scheme including a ground transmitter and a UAV receiver with the help of a multi-antenna terrestrial jammer adopting the zero-forcing technique was studied in \cite{Chen2021a}. Covert communications for UAV-aided data acquisition from multiple ground users were investigated by Zhou \textit{et al.} in terms of improving max-min average covert transmission rate \cite{Zhou2021d}. \textcolor{black}{Shihao \textit{et al.} explored covert communications from the UAV perspective by the joint optimization of UAV's flying location and transmit power for a three-node system model subject to covertness requirement in two-dimensional \cite{Yan2020} and three-dimensional \cite{Zhou2021e} deployment scenarios, and revealed that the latter could achieve better covertness performance than the former. 
It is worth mentioning that some critical limitations of UAV-assisted systems, such as flight power consumption, have not been taken into account in the aforementioned works \cite{Liang, Chen2021, Chen2021a, Zhou2021d, Yan2020, Zhou2021e}.
Indeed, energy efficiency, being an important performance index for realizing green and sustainable wireless networks and IoT systems, should be carefully considered for the energy-constrained UAV-empowered IoT networks.} Overall, the research endeavors about covert communication for dynamic UAV-aided systems are still in the stage of infancy \cite{Jiang2021b}, leaving many opportunities for future developments in various practical scenarios. }

\subsection{Our Contributions}
\textcolor{black}{Inspired by the aforementioned research, in this work, we consider a wireless communication system over THz bands, where an AP intends to communicate with multiple ground UEs in the presence of environmental blockages. We assume that the AP sends confidential data to \textcolor{black}{the scheduled UE per time slot, a.k.a. Bob},  via a UAV-mounted intelligent reflecting surface (UIRS) due to no direct path in-between, while \textcolor{black}{the unscheduled UEs in the same time slot, a.k.a. Willies}, may not be trustworthy. We devise a novel covert communication protocol to guarantee secure data transmission, and our detailed contributions are summarized below.
\begin{itemize}
    \item 
  We devise an energy-efficient multi-UAV secure covert communication scheme to protect information messages and the privacy of the scheduled UE. In particular, we employ one UIRS operating at THz bands for reliable data transmissions from an AP to the scheduled UE per time slot. At the same time,  a \textcolor{black}{UAV-mounted cooperative jammer (UCJ) is utilized} to opportunistically generate AN while benefiting the spacial diversity, aiming at degrading the detection performance of unscheduled UEs Willies. 
  \item
  We obtain exclusive expressions for the minimum detection performance of unscheduled UEs in terms of the missed detection (MD) rate and the false alarm (FA) rate as the most critical metrics for covertness evaluation. A tight lower bound on the average covert data rate from AP to Bob is also derived.
    \item
     To improve covert communication while reducing network power usage, we formulate an optimization problem in terms of a new measure: the minimum average energy efficiency (mAEE). Here the mAEE is defined as the minimum average-ratios between lower-bound covert throughput from AP to the scheduled UE set, and UAVs' propulsion power consumption. However, the optimization problem is nonconvex and challenging to solve optimally. 
    \item 
    To handle this nonconvex problem, we propose a computationally efficient algorithm by applying a  block coordinated successive convex approximation (BSCA) to iteratively solve a  sequence of approximated convex sub-problems such as user scheduling, network power allocation, IRS beamforming optimization, and joint UIRS and UCJ's trajectory and velocity optimization. We then propose a low complex overall algorithm for the system performance improvement with complexity and convergence analysis.
 \end{itemize}
}
The rest of the paper is organized as follows. Section~\ref{sec:sysmodel} introduces our multi-UAV covert communication system and its setting, describing covertness requirements and formulating the energy-efficient design in terms of an optimization problem. In Section~\ref{sec:solution}, we present an iterative low-complex algorithm to solve the optimization problem efficiently, followed by numerical results and discussions in Section \ref{sec:numerical}. Finally, we draw  conclusions in Section~\ref{sec:conclusion}.

{\em Notations}: Throughout this paper, superscripts $(\cdot)^T$ and $(\cdot)^\dagger$ denote transpose and Hermitian transpose operations. The operators $\E\{\cdot\}$ and $\Pr\{\cdot\}$ represent expectation and probability of an even, respectively; $\|\cdot\|$ denotes the Frobenius norm. Also, ${\cal CN}(0,\sigma^2)$ expresses the complex
Gaussian distribution with zero mean and variance $\sigma^2$.
The bold lower-case and upper-case letters denote a vector and matrix, respectively; the upper-case calligraphy letter indicates a set. Define $\mathbb{R}^+$ and $\mathbb{C}$ as the sets of nonnegative real and complex numbers, respectively; $\mathbb{S}^+$ as the set of positive semidefinite (PSD) matrices.

\section{Multi-UAV Covert Communication System model and problem formulation} \label{sec:sysmodel}
\begin{figure}[t]
\centering
\includegraphics[width= 0.9\columnwidth]{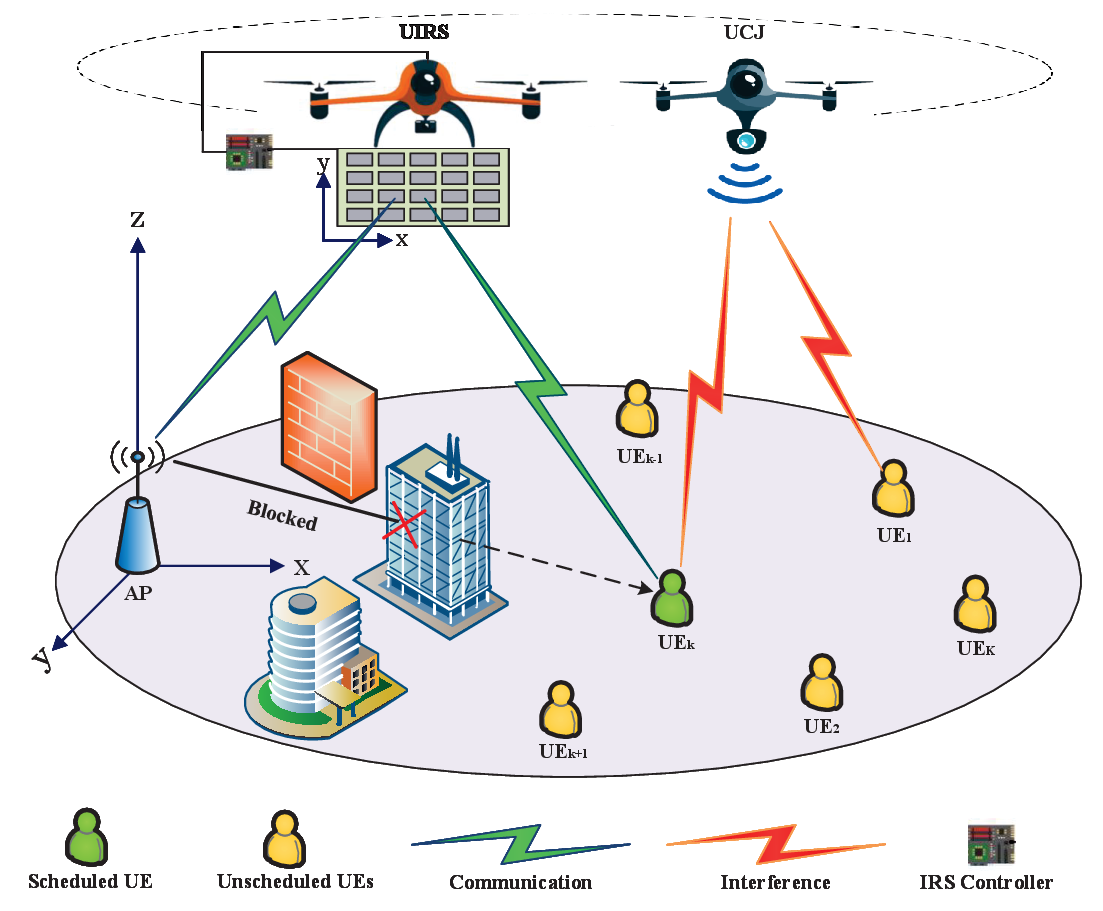}
\caption{System model of cooperative UIRS-UCJ aided THz covert communications in B5G-IoT networks for secure data dissemination.}
\label{fig1:sysmodel}
\end{figure}

\subsection{Multi-UAV Covert Communication System}
We consider a THz-supported wireless communication system as illustrated in Fig. \ref{fig1:sysmodel}, where a UIRS acts as a passive mobile relay to facilitate reliable end-to-end transmissions from an AP towards multiple terrestrial UEs. There is assumed to exist no direct link between  AP and UEs due to severe blockage \cite{Wu2021}. However, UIRS is likely to have strong LoS links with ground terminals due to relatively higher altitude, and mobility \cite{Zhang2019c}. We assume that only one UE (Bob) is scheduled at each time instant $t$, and AP intends to covertly transmit confidential information to Bob to keep the transmission hidden from the unscheduled UEs (Willies). Such a Bob-Willies scenario may arise in large-scale distributed IoT networks where it is difficult to guarantee perfect trustworthiness and transparency of all UEs; thereby, AP needs to adapt the communication protocol according to not only the legitimate UE's requirements, 
but the existence of potential adversaries, e.g., Willies, whom the network operator can identify. In order to assist covert communication, a UCJ is also employed for strategically generating AN, which has been found an effective method to combat warden Willies (see \cite{Soltani2018Covert} and references therein).

\begin{remark}
\textcolor{black}{It should be mentioned that the IRS controller, mounted on the UIRS in our work, acts as a gateway to communicate with other network components (e.g., AP and UEs) through dedicated wireless backhaul/control links. However, in traditional IRS systems, the exchange of control bits can also be accomplished using wired links or fiber channels.}
\end{remark}

\subsection{System Setting Assumptions}
We assume that AP and UCJ are equipped with a single transmit antenna, while all the terrestrial UEs are equipped with a receive antenna for data collection.
Without loss of generality, we consider a 3D Cartesian coordinate system to indicate the location information of each transceiver. We denote AP's location as $\q_a=[x_a,y_a,0]^\mathrm{T}$. We assume there are $K$ randomly distributed terrestrial UEs in the geographical region of interest with the fixed coordinates $\q_{k}=[x_k, y_k,0]^\mathrm{T}$ for $\forall k\in\K$, where $\K=\{1, 2,\cdots, K\}$. Further, due to UAV's limited on-board battery resource, we assume that UIRS and UCJ fly at the fixed altitudes $H_r$ and $H_j$ for a finite period ${T}$, where the altitudes are properly chosen to avoid possible collision with environmental obstacles. This fixed-altitude operation can avoid mechanical energy consumption caused by UAVs' rise and fall \cite{mamaghani2021terahertz,mamaghani2020_tvt}. 

To facilitate the UAVs' trajectory and velocity design, we adopt the time-slotted approach such that the flight duration ${T}$ is equally discretized into $N$ sufficiently small time slots $\delta_t \treq \frac{T}{N}$, where $\delta_t$ should be selected properly to balance between computational complexity and approximation accuracy. Therefore, the UIRS and UCJ's continuous trajectory and velocity sets, denoted as $\{\q_r(t), \v_r(t) \treq\diff[1]{\q_r(t)}{t} \} $ and $\{\q_j(t),  \v_j(t) \treq\diff[1]{\q_j(t)}{t}\}$, for $0 \leq t \leq T$, can be discretized by replacing $t$ by $n\delta_t$, yielding the discrete sets as 
$\{\q_r[n]=[x_r[n],y_r[n], H_r]^\mathrm{T},~\v_r[n]\}$ and $\{\q_j[n]=[x_j[n],y_j[n], H_j]^\mathrm{T},~\v_j[n]\}$, respectively.
 Moreover, we assume that each UE  solely knows the channel distribution information (CDI) between itself and other UEs, while being aware of the channel between itself and the UAVs. In addition, the location information of all UEs is known to the UAVs, since all the UEs are part of the legitimate network serviced by the UAVs in different time slots. 

\subsection{UAVs' Flight Power Model and Mission Requirements}
We consider network power consumption is dominated by the mechanical power consumption of the energy-limited rotary-wing UIRS and UCJ in terms of propulsion for aerial operation, which can be mathematically expressed as \cite{mamaghani2021terahertz}
\begin{align}\label{flightpow_uirs}
    P_{f,r}[n] &=P_o\left(1+c_0\|\v_r[n]\|^2\right) + c_1 \|\v_r[n]\|^3  \nonumber\\
    &\hspace{-10mm}+  P_i\left(\sqrt{1+c^2_2\|\v_r[n]\|^4} -c_2\v_r[n]\|^2\right)^{\frac{1}{2}},\forall n \in \N 
\end{align}
and
\begin{align}\label{flightpow_ucj}
    P_{f,j}[n] &=P_o\left(1+c_0\|\v_j[n]\|^2\right) + c_1 \|\v_j[n]\|^3  \nonumber\\
    &\hspace{-10mm}+  P_i\left(\sqrt{1+c^2_2\|\v_j[n]\|^4} -c_2\v_j[n]\|^2\right)^{\frac{1}{2}},\forall n \in \N 
\end{align}
 wherein $\{\v_r[n], \v_j[n], \forall n\}$ are UAVs' instantaneous velocity in time slot $n$, $P_o$ and $P_i$ are the UAVs' \textit{blade profile power} and \textit{induced power} in hovering mode, respectively, and   $c_0$, $c_1$, and $c_2$ are  some mechanical and environmental-related constants \cite{Zeng2019b}. When a UAV is aloft, i.e., $\|\v_{r(j)}[n]\| = 0$, its mechanical power consumption is $\bar{P}_f = P_i + P_o$ which is not necessarily the minimum power consumption and thereby hovering at a specific point may not be an energy efficient approach for UAV deployment.

\textcolor{black}{Here, we consider that UAVs are deployed to periodically fly over the sky providing covert communications to the UEs; thereby, flight constraints  to the UIRS can be imposed by \begin{subequations}\label{uirs_flight_conds}
\begin{align}
\C1:&\quad\q_r[1] = \q_r[N] = \q^{I}_{r},\label{uirs_flight_cond1}\\
&\quad\|\q_r[n] - \q_a\| \leq \sqrt{R^2_o+H^2_r},~\forall n \in \N\label{uirs_flight_cond2}\\
&\quad \q_r[n+1] = \q_r[n] +\v_r[n]\delta_t,~\forall n \in \N \setminus \{N\} \label{uirs_flight_cond3}\\
 &\quad\|\v_r[n]\|\leq v^{max}_r,~\forall n \in \N\label{uirs_flight_cond4}\\
&\quad \|\v_r[n+1] - \v_r[n]\| \leq a^{max}_r,~\forall n \in \N \setminus \{N\}\label{uirs_flight_cond5}
\end{align}
\end{subequations}
where \eqref{uirs_flight_cond1} is to ensure a periodic flight on the grounds that the UIRS has to return to the initial location by the end of the last time slot, \eqref{uirs_flight_cond2} restricts UIRS's flying region within the \textit{permitted zone}; the horizontal projection of which is assumed to be a circular region with radius $R_o$ in meter centered at AP's location}. Plus, \eqref{uirs_flight_cond3}, \eqref{uirs_flight_cond4}, and \eqref{uirs_flight_cond5} represent UIRS's mobility constraints from practical perspective. Similarly, UCJ's flight constraints can be given by
\begin{subequations}\label{ucj_flight_conds}
\begin{align}
\C2:&\quad\q_j[1] = \q_j[N] = \q^{I}_{j},\\
&\quad\|\q_j[n] - \q_a\| \leq  \sqrt{R^2_o+H^2_j},~\forall n \in \N\\
&\quad\q_j[n+1] = \q_j[n] +\v_j[n]\delta_t,~\forall n \in \N \setminus \{N\}\\
& \quad\|\v_j[n]\|\leq v^{max}_j,~\forall n \in \N\\
&\quad \|\v_j[n+1] - \v_j[n]\| \leq a^{max}_j,~\forall n \in \N \setminus \{N\}
\end{align}
\end{subequations}
where $\q^{I}_{r}$ and  $\q^{I}_{j}$ are UAVs' predetermined stations per flight, $\{v^{max}_r, v^{max}_j\}$ and $\{a^{max}_r, a^{max}_j\}$ are the UAVs' instantaneous maximum speeds and accelerations, respectively. To avoid possible collision in the multi-UAV system, we need to consider safety distance between the UAVs, represented as
\begin{align}\label{safety_cond}
\C3:\quad\|\q_r[n] - \q_j[n]\| \geq D_s,~\forall n\in \N
\end{align}
where $D_s$ denotes the minimum required distance between the two UAVs throughout the periodic mission. 

\subsection{Transmission Strategy}
We assume that direct links between AP and UEs are absent due to severe blockage or considerable distance, necessitating the significance of aerial platforms such as UIRS to establish a reliable wireless link for data dissemination. \textcolor{black}{To support UIRS-assisted downlink covert data transmission service to all the UEs, we employ a time division multiple access (TDMA) protocol as in \cite{Zeng2019b}, wherein the mission time $T$ is divided into $N$ time slots, and at most one UE is scheduled per time slot for intended data transmission, while the unscheduled UEs play as adversaries attempting to detect the presence of communication. This would enable us to fully exploit the time-varying channels with the flexible trajectory design of the considered multi-UAV system.} Thus, by letting $\alpha_k[n]$ be a binary user scheduling variable for UE $k$ in time slot $n$, we have the user scheduling constraint as 
\begin{subequations}\label{usrsch_conds}
\begin{align}
\C4:&\quad\alpha_k[n] \in \{0, 1\},~\forall k\in\K,~\forall n\in \N \label{usrsch_cond1}\\
&\quad\sum^{K}_{k=1} \alpha_k[n] \leq 1,~\forall n\in \N \label{usrsch_cond2}
\end{align}
\end{subequations}
where $\alpha_k[n]=1$ if UE $k$ is scheduled in time slot $n$, and zero otherwise.

\subsection{Channel Modeling}
\textcolor{black}{Since our work examines transmissions over \textcolor{black}{LoS-dominant} THz frequencies, the THz channel gain, encapsulating both large-scale attenuation absorption losses\footnote{\textcolor{black}{
Note that more practical channel modelling and system design for THz propagation is yet to be fully understood, and indeed, requires sophisticated site measurement campaigns and dedicated research efforts, which we leave as a future work \cite{Tekbiyik2020}.}}, for direct links from UCJ to UE $k$ in time slot $n$ can be denoted, similar in \cite{Wang2020, mamaghani2021terahertz, Xu2020, Boulogeorgos2018}, as
\begin{align}
    h_{jk}[n] &= \left(\frac{C}{4\pi f_c \|\q_j[n] - \q_k \|^{\frac{\rho}{2}}}\right)
    \exp\left(\frac{-j2\pi \|\q_j[n] - \q_k \|}{\lambda_c}\right) \nonumber\\
    &\hspace{-5mm}\times\exp\left(-\frac{\kappa(f_c, \mu) \|\q_j[n] - \q_k \|}{2}\right),~\forall k\in\K, n\in\N
\end{align}
wherein $\lambda_c \treq \frac{C}{f_c}$ is the transmission wavelength in meter, $C\approx 3\times 10^8~\SI{}{\meter\per\second} $ is the speed of light, and $f_c$ specifies the carrier frequency, \textcolor{black}{$\rho$ determines the large-scale pass-loss exponent which usually satisfies $2\leq \rho\leq 4$, ranging between free space and obstructed propagation environments \cite{Goldsmith2005b}}. Furthermore, $\kappa(f_c, \mu)$ represents the overall molecular absorption coefficient of THz channels as a function of $f_c$ and the volume of the mixing ratio of water vapor $\mu$, describing the relative absorbing area of the molecules in the wireless medium per unit volume. It should be stressed that the main cause of absorption loss in THz frequency ranges is the water vapor molecules that cause a discrete but deterministic loss to the signals in the frequency domain \cite{Tekbiyik2020, Boulogeorgos2018}. }

\begin{figure}[t]
\centering
\includegraphics[width= \columnwidth]{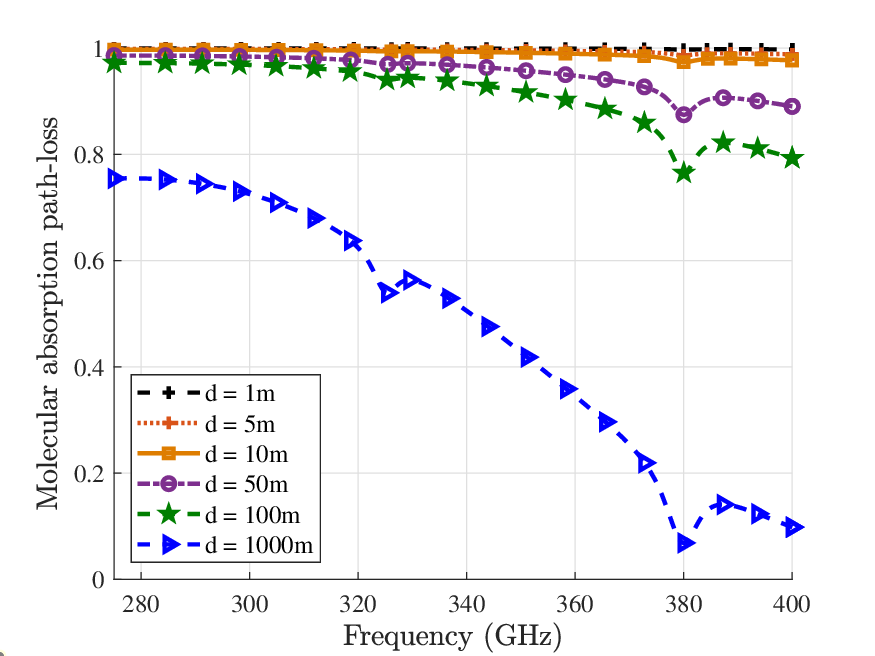}
\caption{\small Illustration of how THz molecular absorption path-loss varies with carrier frequency and distance ($d$) for standard atmospheric condition.}
\label{remark:fig1}
\end{figure}

\begin{remark}
The molecular absorption loss in $\SI{275}-\SI{400}{\giga\hertz}$ can be modelled, according to \cite{Boulogeorgos2018}, as $L^{-1}_a(f_c,d,\mu) = \exp(-\kappa(f_c, \mu)d)$ where $d$ is the distance, and 
\begin{align}\label{kappa_eq}
   \kappa(f_c, \mu) &= \frac{0.2205\mu(0.133\mu\hspace{-0.5mm}+\hspace{-0.5mm}0.0294)}{\left(0.4093\mu\hspace{-0.5mm}+\hspace{-0.5mm}0.0925\right)^2 \hspace{-0.5mm}+\hspace{-0.5mm} \left(\frac{f_c}{100C}\hspace{-0.5mm}-\hspace{-0.5mm}10.835\right)^2}\nonumber\\
   &\hspace{-12mm}\hspace{-0.5mm}+\hspace{-0.5mm} \frac{2.014\mu(0.1702\mu\hspace{-0.5mm}+\hspace{-0.5mm}0.0303)}{\left(0.537\mu\hspace{-0.5mm}+\hspace{-0.5mm}0.0956\right)^2 \hspace{-0.5mm}+\hspace{-0.5mm} \left(\frac{f_c}{100C}\hspace{-0.5mm}-\hspace{-0.5mm}12.664\right)^2}\hspace{-0.5mm}+\hspace{-0.5mm} {5.54\times10^{\hspace{-0.5mm}-\hspace{-0.5mm}37}} f^3_c\nonumber\\
   &\hspace{-12mm} {-3.94\times10^{-25}} f^2_c \hspace{-0.5mm}+\hspace{-0.5mm} {9.06\times10^{-14}}f_c \hspace{-0.5mm}-\hspace{-0.5mm} {6.36\times10^{-3}},
\end{align}
in which  $\mu$ can be evaluated as
\[
    \mu = 6.1121({3.46\hspace{-0.5mm}\times\hspace{-0.5mm}10^{-8}} P \hspace{-0.5mm}+\hspace{-0.5mm} 1.0007)\frac{\phi}{P}\exp\left(\frac{17.502 T}{240.97 \hspace{-0.5mm}+\hspace{-0.5mm} T}\right)
\]
wherein $\phi$ stands for the relative humidity in percentage, $P$ denotes the pressure in $\SI{}{\pascal}$, and $T$ is measured in $\SI{}{\degreeCelsius}$. The THz link molecular absorption path-loss model given above and illustrated in Fig. \ref{remark:fig1} was shown to have high accuracy for up to $1$ km links in standard atmospheric conditions, i.e., the temperature of \SI{296}{\kelvin} and pressure of \SI{101325}{\pascal}. However, it is worth mentioning that nonstandard atmospheric conditions can also be described via this model.
\end{remark}

On the other hand, the cascaded channel gain of the AP-UIRS-UE $k\in\K$ in time slot $n\in\N$ can be represented, in the far field scenario, as \cite{Tang2021, Pan2021}
\begin{align}
    \tilde{h}_{ark}[n]&= \left(\frac{C}{8\pi\sqrt{\pi}f_c \|\q_r[n] - \q_a \|^{\frac{\rho}{2}}\|\q_r[n] - \q_k \|^{\frac{\rho}{2}}}\right)\nonumber\\
    &\hspace{-12mm}\times\exp\left(-\frac{j2\pi (\|\q_r[n] - \q_a \| + \|\q_r[n] - \q_k \|)}{\lambda_c}\right)\nonumber\\
    &\hspace{-12mm}\times\exp\left(-\frac{\kappa(\|\q_r[n] - \q_a \| + \|\q_r[n] - \q_k \|)}{2}\right),
\end{align}
where $\kappa \treq \kappa(f_c, \mu)$ for notation simplicity.
Consider there is an IRS as a uniform planar array (UPA) deployed at UIRS parallel to the ground. Let $L_x$ and $L_y$ be  the number of reflecting elements alongside the $\mathrm{x}$ and $\mathrm{y}$-axes of the IRS (see Fig. \ref{fig1:sysmodel}), respectively, so the total number of reflecting elements is $L = L_xL_y$. 
Here, we assume that the 3D coordinate of the first element of the IRS (the IRS element shown at the origin in Fig. \ref{fig1:sysmodel}) equals to the instantaneous location of UIRS, i.e., $\q_r[n]$. Therefore, the transmission vector from AP towards the first element of the IRS is $(\q_r[n]-\q_a)$, and also, the difference vector from the IRS can be represented as $\Delta \r_{l_x,l_y} = [(l_x-1)\delta_x, (l_y-1)\delta_y, 0]^\mathrm{T}$, where $l_x$ and $l_y$ represent the $l_x$-th row and $l_y$-th column of the IRS, $\delta_x$ and $\delta_y$ denote the element separation alongside $\mathrm{x}$ and $\mathrm{y}$ axes. Accordingly, the relative phase difference between the signal received by the first element and the $(l_x,l_y)$-th element of the IRS in time slot $n$ can be represented by
\begin{align}
    \theta_{l}[n] = \frac{2\pi(\q_r[n]-\q_a)^\mathrm{T} \Delta r_{l_x,l_y}}{\lambda_c\|\q_r[n] - \q_a\|},~\forall n\in \N
\end{align}
where the subscript $l = (l_x-1)\times L_y + l_y$. Henceforth, the received array vector from AP to UIRS is given by
\begin{align}
    \e_a[n] = [\mathrm{e}^{-j\theta_{l}[n]}, \cdots, \mathrm{e}^{-j\theta_{L}[n]}]^\mathrm{T},~\forall n\in \N
\end{align}
Similarly, the relative phase difference between the first and $(l_x, l_y)$-th element of the IRS's reflected beams towards UE $k$ is given by
\begin{align}
    \beta^{k}_{l}[n] = \frac{2\pi(\q_k-\q_r[n])^\mathrm{T} \Delta r_{l_x,l_y}}{\lambda_c\|\q_k -\q_r[n]\|},~\forall k\in \K,~n\in \N
\end{align}
Therefore, the transmit array beam from UIRS to the $k$-th UE in time slot $n$ can be represented as
\begin{align}
    \e_k[n] = [\mathrm{e}^{-j\beta^{k}_{1}[n]}, \cdots, \mathrm{e}^{-j\beta^{k}_{L}[n]}]^\mathrm{T},~\forall k\in \K,~n\in \N
\end{align}
Hence, we can express the channel gain of AP-UIRS-UE $k$ as
\begin{align}\label{ap-uirs-ue_ch_gain}
    h_{ark}[n] =   \e^\dagger_k[n] \Phi[n] \e_a[n] \tilde{h}_{ark}[n],~\forall k\in \K,~n\in \N
\end{align}
where $\Phi[n]$ is an $L$-by-$L$ beamforming matrix of the IRS defined as
\[
    \Phi[n] \treq \mydiag \left(\rho_1[n]\mathrm{e}^{j\phi_1[n]}, \cdots, \rho_L[n]\mathrm{e}^{j\phi_L[n]} \right)
\]
each element of its main diagonal represents both the amplitude and the phase shift of the $(l_x,l_y)$-th element of the IRS in time slot $n$ with the constraints given by
\begin{subequations}\label{beamforming_conds}
\begin{align}
\C5:&\quad0 \leq \rho_l[n] \leq 1,~\forall l\in\L, n\in\N \label{beamforming_cond1}\\
&\quad0 \leq \phi_l[n] \leq 2\pi,~\forall l\in \L, n\in\N \label{beamforming_cond2}
\end{align}
\end{subequations}
where $\L=\{1, 2, \cdots, L\}$.

\subsection{Signals Representation}

\textcolor{black}{In the $n$-th time slot, AP communicates with the designated UE by mapping the intended message to the sequence $\mathbf{x}_a[n]=[x^1_a[n], x^2_a[n],\cdots, x^\aleph_a[n]]$, where $\aleph$ is the total number of channel uses per time slot, and then conducting transmission via the UIRS. It is also assumed that the transmission of a message to any UE is completed within one time slot, the period in which channel properties remain constant. Also, slot boundaries are shared and synchronized amongst the communication nodes. In addition, the average power per symbol in $\mathbf{x}_a[n], \forall n$ is normalized to unity. Further, in this work, Gaussian signalling\footnote{\textcolor{black}{It has been proved that Gaussian signalling under the AWGN channels is optimal for covert communications in terms of maximizing the mutual information of transmitted and received signals; however, it may not be optimal if different covertness constraint is adopted \cite{Yan2019a}.}}, wherein the transmitted signal follows a normal distribution, is employed by the AP with zero-mean and variance, a.k.a. transmit power, $p_a[n], \forall n$. Thus, the signal vector received by the $k$-th UE  can be represented by
\begin{align}\label{yk}
   \mathbf{y}_k[n] &= \sqrt{p_a[n]}h_{ark}[n] \mathbf{x}_a[n]   \nonumber\\
   &+ \sqrt{p_j[n]}h_{jk}[n] \mathbf{x}_j[n] + \pmb{\delta}_k[n],~\forall k\in \K,~n\in\N
\end{align}
where $\pmb{\delta}_k[n]\in\mathcal{CN}(\mathbf{0},\sigma^2_k\mathbf{I}_\aleph)$ is the UE $k$'s receiver noise vector, modelled as the additive white Gaussian noise (AWGN), whose elements follow zero mean and variance $\sigma^2_k$; $\mathbf{x}_j[n]$ represents the AN vector transmitted by UCJ following $\E\{\|\mathbf{x}_j[n]\|^2\}=1$.}
Further, $p_j[n]$  denotes the UCJ's transmit power in time slot $n$. In this work, we assume that $p_j[n], \forall n$ is a random variable\footnote{\textcolor{black}{In this work, we consider an uninformed jammer case \cite{Sobers2017}, wherein AP and UCJ are not closely coordinated. Otherwise, AP can generate codeword symbols independently from the Gaussian jamming distribution, providing it solely to the scheduled UE as a shared secret. Then, while AP starts to transmit a signal to the intended UE, UCJ reduces its jamming transmission and then turns it back up once the AP's transmission is done. By doing such, Willies are unable to determine any change has occurred during the course of transmission.}} with a uniform distribution in the interval $[0, \hat{p}_j[n]],~\forall n$, wherein $\hat{p}_j[n]$ denotes the peak AN transmission power by UCJ in time slot $n$, with the following   probability density function (pdf) and cumulative distribution function (cdf) as
\begin{align}\label{pdf_pj}
    f_{p_j[n]}(x) = \begin{cases} \frac{1}{\hat{p}_j[n]},& 0 \leq x\leq \hat{p}_j[n]\\
                0, & \mathrm{o.w.}
    \end{cases}
\end{align}
\begin{align}\label{cdf_pj}
    F_{p_j[n]}(x)= \begin{cases} 
                0,&  x \leq 0 \\
                \frac{x}{\hat{p}_j[n]},& 0 \leq x\leq \hat{p}_j[n]\\
                1, &  x \geq \hat{p}_j[n] 
    \end{cases}
\end{align}
with expected value $\E\{p_j[n]\} = \frac{\hat{p}_j[n]}{2}$.

It is worth pointing out that AP's transmit power and the peak AN power of UCJ are generally subject to maximum instantaneous and total network power budget constraints represented as
\begin{subequations}\label{pow_conds}
\begin{align}
\C6:&\quad0 \leq p_a[n] \leq p^{max}_a,~\forall n\in\N\\
&\quad 0 \leq \hat{p}_j[n] \leq p^{max}_j,~\forall n \in \N \label{pow_cond1}\\
&\quad\sum_{n=1}^{N}p_a[n] + \hat{p}_j[n] \leq {p}^{tot},\label{pow_cond2}
\end{align}
\end{subequations}

\begin{remark}
\textcolor{black}{By introducing randomness in the UCJ's AN transmissions, we can create ambiguity in the received power at Willies (unscheduled UEs) to assist covert transmissions to Bob (the scheduled UE). 
From a conservative point of view, we assume that Willies know the distribution information of the UCJ's AN transmit powers as well as their AWGN noise variances, which is the worst-case scenario since it becomes much easier for them to make a decision less erroneous than the cases without such information.}
\end{remark}

Following \eqref{yk}, the average channel capacity from AP to UE $k$ taking over the randomness nature of AN transmission powers by UCJ can be obtained as
\begin{align}\label{rk_lb}
    \bar{R}_{k}[n] &= \E_{p_j[n]}\left\{\mathrm{W}\log_2\left(1+\frac{p_a[n] g_{ark}[n]}{p_j[n]g_{jk} + \sigma^2_k[n]}\right)\right\},\nonumber\\
    &\hspace{-10mm}\stackrel{(a)}{\geq} \mathrm{W}\log_2\left(1+\frac{p_a[n] g_{ark}[n]}{\frac{1}{2}\hat{p}_j[n]g_{jk}[n] + \sigma^2_k[n]}\right) \treq \bar{R}^{lb}_{k}[n],
\end{align}
wherein $\mathrm{W}$ is the allocated transmission bandwidth\footnote{\textcolor{black}{Notice that to confront the high frequency-selectivity nature in the THz band, the total bandwidth of the THz frequencies is generally divided into several sub-bands \cite{mamaghani2021terahertz, Pan2021}. Here, we consider only one sub-band, equally shared amongst communication nodes with the associated carrier frequency.}} in Hz, $g_{ark} \treq \|{h_{ark}[n]}\|^2$, $g_{jk}[n] \treq \|{h_{jk}[n]}\|^2$. Further, $(a)$ follows from applying Jensen's inequality theorem wherein given $X$ be an arbitrary random variable then according to \cite{Boyd2004} we have 
\[\E\{f(X)\} \geq f(\E\{X\}) \iff f(X)~\text{is convex} \]
As a result, owing to the convexity of the function $\log(1+\frac{a}{bx+1})$ with respect to (w.r.t) the variable $x$ in the domain of $x\geq -\frac{2+a}{2b}$ for $a, b\geq 0$ \cite{Mamaghani2021}, the tight lower-bound expression $\bar{R}^{lb}_{k}[n]$ can be obtained.

\subsection{Covert Communication Requirement and Analysis}
For the non-colluding Willies, i.e., each of the unscheduled UEs (a.k.a Willies) $\W \treq \K\setminus\{k\}$ independently attempts for conducting malicious activity in terms of signal transmission detection based on their own observations of the given block of transmissions\footnote{\textcolor{black}{Note that here channel properties varies independently from one block of transmission (time slot) to the others; therefore, the knowledge of previous blocks cannot help Willies improve their detection performance.}}. Thus, each of Willies encounters a  binary hypothesis testing problem to independently decide whether AP has transmitted information signal towards the scheduling UE $k$. This non-colluding scenario is valid as the UEs are randomly distributed in the region, and each potentially serves as a scheduled UE in some specific time slots during the mission. Therefore, the received signal vector at the $m$-th Willie in time slot $n$ can be represented as
\begin{align}\label{recsignWillie}
    \mathbf{y}_m[n] =\begin{cases}
      \sqrt{p_j[n]}h_{jm}[n] \mathbf{x}_j[n] + \pmb{\delta}_m[n],& \H_0 \\
        \begin{aligned}
        &{\sqrt{p_a[n]}}h_{arm}[n] \mathbf{x}_a[n]\\
            &+{\sqrt{p_j[n]}}h_{jm}[n] \mathbf{x}_j[n] + \pmb{\delta}_m[n],
        \end{aligned} & \H_1
    \end{cases}
\end{align}
where $\pmb{\delta}_m[n]\sim \mathcal{CN}(\mathbf{0}, {\sigma}^2_m[n]\mathbf{I}_\aleph)$ is the AWGN vector at the receiver of Willie $m \in \W$; $\mathcal{H}_0$ is the null hypothesis stating that AP has not transmitted information signal, whereas $\mathcal{H}_1$ is the alternative hypothesis. 

\textcolor{black}{In this work, considering somewhat worst-case scenario, we assume that Willies have complete statistical
knowledge of their observations. As such,  the parameters for AP’s random codeword generation, UCJ's random AN jamming, the noise variance vector $\pmb{\sigma}^2_m[n],~\forall n$, and the location information of both the UAVs and the AP are known to all Willies.
Accordingly, Willie $m$ has to decide between the hypotheses $\mathcal{H}_0$ and $\mathcal{H}_1$ considering AP's transmissions towards the intended UE. Therefore, one can apply the Neyman-Pearson (NP) criterion in order to obtain the optimal test for Willies to minimize their detection error rate via utilizing the \textit{likelihood ratio test} (LRT) given as 
\begin{align}
    \Lambda(\mathbf{y}_m) = \frac{f_{{\mathbf{y}_m} | {p_j[n],  \mathcal{H}_1}}(\mathbf{y}_m | p_j[n], \mathcal{H}_1)}{f_{{\mathbf{y}_m}|{p_j[n], \mathcal{H}_0}}(\mathbf{y}_m|p_j[n], \mathcal{H}_0)}
    \underset{\D_0}{\overset{\D_1}{\gtrless}} \gamma,
\end{align}
where $\gamma \treq \frac{\Pr\{\mathcal{H}_1\}}{\Pr\{\mathcal{H}_0\}}$, and following the assumption of equal \textit{a priori} probability of each  hypothesis being true, we have $\gamma=1$.
Further, $\D_0$ and $\D_1$ denote the binary decisions in favor of hypothesis $\mathcal{H}_0$ and $\mathcal{H}_1$, respectively, $f_{{\mathbf{y}_m}|{p_j[n], \mathcal{H}_1}}(\mathbf{y}_m|p_j[n], \mathcal{H}_1)$ and $f_{{\mathbf{y}_m}|{p_j[n], \mathcal{H}_0}}(\mathbf{y}_m|p_j[n], \mathcal{H}_0)$ are the likelihood functions of the $m$th Willie's observation vector.}

\textcolor{black}{It has been shown, using the concepts of stochastic ordering \cite{Moshe}, that detection using a radiometer is indeed the optimal decision rule for Willies in the considered system model that minimizes their detection error \cite{Yan2019, Shahzad2017, Zhoua}. Therefore, we assume that Willies each use a radiometer for signal detection, conducting a threshold test on the average power received, similar to \cite{Hu2020, Sobers2017, Zhou2021d, Zhou2021e}.
Accordingly, the adopted detector's decision rule at Willie $m$ in time slot $n$  can be rewritten as
\begin{align}\label{optimaldecision}
    \Xi[n] \treq \frac{\Gamma_m[n]}{\aleph}  \underset{\D_0}{\overset{\D_1}{\gtrless}} \varrho_m[n],
\end{align}
where $\Gamma_m[n]  \treq \sum^{\aleph}_{j=1}\|y^{(j)}_m[n]\|^2,~\forall m, n$ is the total received power at Willie $m$ in time slot $n$, and $\varrho_m[n]$ represents the corresponding detection threshold, which will be optimized later for minimizing the total detection error probability.
} Here we adopt the widely used infinite blocklength assumption in covert communications, i.e., $\aleph \longrightarrow \infty$, implying that each Willie can observe an infinite number of samples, which is, in fact, an upper bound on the number of received samples in practice. Nonetheless, owing to the fact that the number of symbols transmitted per time slot gets increased with the communication bandwidth and thanks to the abundance of available bandwidth in THz frequencies, this assumption of approximately infinite channel uses per slot can be well justified. Therefore, the expression $\Xi[n],~\forall n$  defined in \eqref{optimaldecision} can be simplified as
\begin{align}
    \Xi[n] = \begin{cases}
    p_j[n]g_{jm}[n] + \sigma^2_m[n], & \H_0\\
    p_a[n]g_{arm}[n] + p_j[n]g_{jm}[n] + \sigma^2_m[n], & \H_1\\
    \end{cases}
\end{align}

\textcolor{black}{Having established the best strategy for Willies to detect confidential communication, we now analytically obtain the best setting for the radiometer's threshold adopted by Willies, whose ultimate goal is to detect their observations produced by which one of the hypotheses $\mathcal{H}_0$ and $\mathcal{H}_1$. In light of this, we adopt the total detection error probability comprised of the MD and FA probabilities to measure Willies' detection performance of AP's transmissions.
}

\subsubsection*{False Alarm Rate}
If Willie $m$ decides that AP has sent data to the scheduled UE $k$ in the $n$-th time slot, while AP has not sent any data, i.e., $\mathcal{H}_0$ is true, we have \textit{False Alarm} (FA) occurrence  with probability of \begin{align}\label{prob_fa}
    P^{FA}_m[n] &\treq \Pr\{\D_1| \H_0\} \nonumber\\
           &= \Pr\left\{p_j[n]g_{jm}[n] + \sigma^2_m[n] \geq \varrho_m[n] \right\},
\end{align}

\subsubsection*{Missed Detection Rate}
If Willie $m$ decides that AP has not transmitted data to the scheduled UE $k$, while $\mathcal{H}_1$ is true, then we say that a \textit{Missed Detection} (MD) incident has occurred with the probability given by
\begin{align}\label{prob_md}
        \mathrm{P}^{MD}_{k,m}[n] &= \Pr\{\D_0| \H_1\} \nonumber\\
        &\hspace{-10mm}=  \Pr\left\{p_a[n]g_{arm}[n] \hspace{-1mm}+\hspace{-1mm} p_j[n]g_{jm}[n] \hspace{-1mm}+\hspace{-1mm} \sigma^2_m[n] \hspace{-1mm} \leq\hspace{-1mm} \varrho_m[n]\right\},
\end{align}
Thus, the total detection error rate at Willie $m$ given $k$-th Bob in time slot $n$ is expressed by
\begin{align}\label{det_err}
    \zeta_{k,m}[n] =   P^{FA}_m[n] +  \mathrm{P}^{MD}_{k,m}[n],~\forall n, k, m
\end{align}

In general, Willie $m$ attempts to minimize the detection error rate in \eqref{det_err}, while the UIRS-UCJ aims at ensuring this minimum error detection, denoted by $\zeta^\star_{m,k}[n]$, and obtained by solving
\begin{align}
    \zeta^\star_{m,k}[n] = \underset{\varrho_m[n]}{\mathrm{minimize}}~~\zeta_{k,m}[n],~\forall n, k, m
\end{align}
being no more than some specific value at every Willie in time slot $n$, i.e., $\zeta^\star_{m,k}[n] \geq 1- \varepsilon$. It is worth mentioning that  $\varepsilon$ is typically a small, non-negative constant denoting the covertness requirement for data dissemination in the system. In the following, we first derive analytical expressions for the FA and MD probabilities, based on which we then obtain optimal detection threshold from Willies' perspective.

Exploiting the distribution of random variables $p_j[n],~\forall n\in\N$, given by \eqref{pdf_pj} and \eqref{cdf_pj},  we can rewrite \eqref{prob_fa} as
\begin{align}\label{fa_analytical}
    P^{FA}_m[n] &= 1- F_{p_j[n]}\left(\frac{\varrho_m[n] - \sigma^2_m[n]}{g_{jm}[n]}\right)\nonumber\\
    &=\begin{cases}
    1,&  \varrho_m[n] \leq \sigma^2_m[n] \\
    1- \frac{\varrho_m[n] - \sigma^2_m[n]}{\hat{p}_j[n]g_{jm}[n]},&    \sigma^2_m[n] < \varrho_m[n] \leq \chi_2\\
    0,&  \varrho_m[n] >  \chi_2
    \end{cases}
\end{align}
wherein $\chi_2 \treq \hat{p}_j[n]g_{jm}[n]+\sigma^2_m[n]$. Similarly, we can analytically calculate \eqref{prob_md} as
\begin{align}\label{md_analytical}
    P^{MD}_{k,m}[n] &= F_{p_j[n]}\left(\frac{\varrho_m[n] - p_a[n]g_{arm}[n] - \sigma^2_m[n]}{g_{jm}[n]}\right)\nonumber\\
    &=\begin{cases}
    0,&  \varrho_m[n] \leq \chi_1 \\
    \frac{\varrho_m[n] - p_a[n]g_{arm}[n] - \sigma^2_m[n]}{\hat{p}_j[n]g_{jm}[n]},&    \chi_1  < \varrho_m[n] \leq  \chi_3\\
    1,&  \varrho_m[n] >  \chi_3
    \end{cases}
\end{align}
wherein $\chi_1 \treq p_a[n]g_{arm}[n] + \sigma^2_m[n]$ and $\chi_3 \treq p_a[n]g_{arm}[n] + \hat{p}_j[n]g_{jm}[n]+\sigma^2_m[n]$. We note that if $\hat{p}_j[n]g_{jm}[n] \leq p_a[n]g_{arm}[n]$, then by setting $\varrho_m[n] = p_a[n]g_{arm}[n]+\sigma^2_m[n]$, Willie $m$ can achieve zero detection error rate. Hence, we consider the non-trivial case wherein  $\chi_1 \leq \chi_2 \leq \chi_3$ which is equivalent with  
\[
\hat{p}_j[n]g_{jm}[n] \geq p_a[n]g_{arm}[n],~\forall n\in \N,~m\in\W
\]
Hence, the following analyses are based on the aforementioned critical assumption for covertness requirement. Now, we can compute $\zeta_{k,m}[n]$ given by \eqref{det_err}, as
\begin{align}
     \zeta_{k,m}[n] = \begin{cases}
      1,&  \varrho_m[n] \leq \sigma^2_m[n] \\
      1- \frac{\varrho_m[n] - \sigma^2_m[n]}{\hat{p}_j[n]g_{jm}[n]},& \sigma^2_m[n] \leq \varrho_m[n] < \chi_1\\
      1- \frac{ p_a[n]g_{arm}[n] }{\hat{p}_j[n]g_{jm}[n]},& \chi_1 \leq \varrho_m[n] < \chi_2\\
      \frac{\varrho_m[n] - p_a[n]g_{arm}[n] - \sigma^2_m[n]}{\hat{p}_j[n]g_{jm}[n]},& \chi_2 \leq \varrho_m[n] < \chi_3\\
      1,& \varrho_m[n] \geq \chi_3
  \end{cases}
\end{align}
It is evident that Willies will not set the decision threshold lower than noise variance $\sigma^2_m[n]$ or higher than  $\chi_3$ to ensure that the resulting detection error probability will be less than $1$; otherwise, the detection performance would be the same as of a random guess. Further, we can see that $\zeta_{k,m}[n]$ is a decreasing function in the range $\sigma^2_m[n] \leq \varrho_m[n] < \chi_1$ and increasing function in the range $\chi_2 \leq \varrho_m[n] < \chi_3$ w.r.t  $\varrho_m[n]$ and also behaves as a constant function when $\chi_1 \leq \varrho_m[n] < \chi_2$. Therefore, considering that  $\zeta_{k,m}[n]$ is a continuous function of $\varrho_m[n]$, the optimal decision threshold to be set by Willie $m$ should be in the range $\chi_1 \leq \varrho_m[n] \leq \chi_2$, resulting in the minimum detection error rate 
\begin{align}\label{zeta_min}
   \zeta^\star_{m,k}[n] = 1\hspace{-1mm}-\hspace{-1mm} \frac{ p_a[n]g_{arm}[n] }{\hat{p}_j[n]g_{jm}[n]},~\forall n\in\N,~k\in\K,m\in \W
\end{align}
Then, the covert communication constraint can be stated as
 \begin{align}\label{covert_cond}
\hspace{-3mm}\C7:~\sum^{K}_{k=1} \alpha_k[n] \min_{m \in \W} \zeta^\star_{m,k}[n] \geq 1-\varepsilon,~ \forall n\in\N
 \end{align}
 wherein $\zeta^\star_{m,k}[n]$ is in \eqref{zeta_min}.

 \begin{remark}
\textcolor{black}{It is worth pointing out that, as per the minimum detection error rate obtained in \eqref{zeta_min}, $\zeta^\star_{m,k}[n]$ decreases with the AP's transmit power as well as the downlink channel quality but increases when UCJ's maximum AN transmission power gets increased or the quality of the interference links improves. But on the other hand, the aforementioned parameters have reverse impacts on the covert throughput metric given by \eqref{rk_lb}, reinforcing the inherent trade-off between the covertness requirement and the achievable transmission rate of the considered system. Therefore, this requires us to carefully design the UAVs' trajectory and the communication resources to effectively balance transmission quality and communication covertness.}
 \end{remark}

\section{Problem Formulation and Proposed Low-Complex Solution}\label{sec:solution}
To devise an energy-efficient UIRS-assisted covert communication system, we first formulate the optimization problem aiming to improve the minimum average energy efficiency (mAEE) of the network. Here the mAEE is defined as the minimum average-ratios of lower-bound throughput given by \eqref{rk_lb} to the  UAVs' total propulsion power consumption\footnote{\textcolor{black}{From the practical perspective, mechanical power consumption of the UAVs for hovering or supporting their mobility is the dominant power consumption of the network compared to the communication-related power consumption required for transmissions, signal processing, circuitry, and so forth (e.g., in the order of hundreds of watts versus a few watts \cite{Zeng2017}).}} as 
\begin{align}\label{opt_prob}
(\P):& \stackrel{}{\underset{\pmb{\alpha}, \mathbf{P}, \mathbf{\Phi},\mathbf{Q}_r, \mathbf{Q}_j}{\mathrm{maximize}}~~\underset{k\in\K}{\min}~\frac{1}{N}\sum^{N}_{n=1}\frac{\alpha_k[n] \bar{R}^{lb}_{k}[n]}{P_{f,r}[n] + P_{f,j}[n]}  } \nonumber\\
&~~\mathrm{s.t.}~~~~~ \C1-\C7,
\end{align}

Note that problem $(\P)$ is a mixed-integer fractional nonconvex non-linear programming, which is challenging to solve optimally. Indeed, the major challenge in solving $(\P)$ arises from the binary user scheduling constraint $\C4$, nonconvex constraints $\C3$ and $\C7$, and the highly coupled optimization variables in the fractional-form objective function.  To embark on the non-convexity and make the problem tractable, we propose a computationally efficient algorithm by applying a block coordinated successive convex approximation (BSCA) to iteratively solve a sequence of approximated convex sub-problems by employing several techniques. Specifically, we split problem $(\P)$ into the following sub-problems with different blocks of variables: $i$) \textit{user scheduling sub-problem} to optimize $\pmb{\alpha}$, $ii$) network transmission power sub-problem to optimize $\Pow=\{\Pow_a$, $\hat{\Pow}_j\}$, $iii$) beamforming matrices sub-problem to improve $\mathbf{\Phi}=\{\Phi[n], \forall n\in \N\}$, $iv$) UIRS's joint trajectory and velocity sub-problem to improve $\mathbf{Q}_r=\{\q_r[n], \v_r[n],~\forall n\in\N\}$ and $v$) UCJ's joint trajectory and velocity sub-problem to improve $\mathbf{Q}_j=\{\q_j[n], \v_j[n],~\forall n\in\N\}$. Next, we solve each of them while keeping the other blocks fixed, then propose an overall low-complex algorithm to iteratively attain the approximate solution of \eqref{opt_prob}.

\subsection{Sub-problem I: User Scheduling Optimization}\label{P1_solution}
By keeping the optimization blocks $\mathbf{P}, \mathbf{\Phi}, \mathbf{Q}_r, \mathbf{Q}_j$ fixed and relaxing the binary constraint \eqref{usrsch_conds} into a continuous constraints, we can rewrite $(\P1)$ equivalently as
\begin{subequations}\label{usrsch_subprob}
\begin{align}
(\P1):& \stackrel{}{\underset{\pmb{\alpha}}{\mathrm{maximize}}~~\underset{k\in\K}{\min}~\frac{1}{N}\sum^{N}_{n=1} A_{k,n}\alpha_k[n]}   \nonumber\\
&~~\mathrm{s.t.}~~~~~ \sum^{K}_{k=1} B_{k,n}\alpha_k[n] \geq 1-\varepsilon,~\forall n\in \N \label{P1_cond1}\\
&\quad\quad\quad 0 \leq \alpha_k[n] \leq  1,~\forall k\in\K, n\in\N \label{P1_cond2}\\
&\quad\quad\quad \sum^{N}_{n=1}\sum^{K}_{k=1}\left(\alpha_k[n] - \alpha^2_k[n]\right) \leq  0, \label{P1_cond3}\\
&\quad\quad\quad \sum^{K}_{k=1}\alpha_k[n] \leq  1,~\forall n\in \N \label{P1_cond4}
\end{align}
\end{subequations}
wherein 
\[A_{k,n} = \frac{\bar{R}^{lb}_{k}[n]}{P_{f,r}[n] + P_{f,j}[n]},~~~\text{and}~~~B_{k,n}=\underset{m\in\W}{\min } \zeta^\star_{m,k}[n]\quad\forall k, n\]
We note that considering constraints \eqref{P1_cond2} and \eqref{P1_cond3} jointly ensures that $\alpha_k[n]=0$ or $\alpha_k[n]=1$ must hold, similar to the original binary constraint \eqref{usrsch_conds}, i.e., $\alpha_k[n] \in \{0, 1\},~\forall k,n$. Following the above transformation, though the NP-hard mixed-integer user scheduling optimization problem is equivalently converted into a continuous optimization problem, it is still challenging to solve due to nonconvex linear-minus-quadratic constraint \eqref{P1_cond3}. Here, we apply the first order restrictive approximation \cite{Boyd2004} to obtain a global upper-bound at the given local point $\alpha^{lo}_{k}[n], \forall n, k$ and rewrite it as
\begin{align}\label{usrschi_2_cvx}
    \sum^{N}_{n=1}\sum^{K}_{k=1}\left[(1- 2\alpha^{lo}_{k}[n]) \alpha_k[n] +(\alpha^{lo}_{k}[n])^2\right] \leq  0,
\end{align}

However that replacing \eqref{P1_cond3} by the convex constraint obtained in \eqref{usrschi_2_cvx} converts problem $(\P1)$ into a convex optimization problem which can be sequentially solved via SCA method, due to co-existence of constraints \eqref{P1_cond2} and \eqref{P1_cond3}, it is, in general, difficult to attain a feasible solution. Therefore, to tackle this issue, we apply the penalty-SCA (PSCA) technique wherein we bring the constraint \eqref{usrschi_2_cvx} into the objective function via adding a penalty term. Although this violates the binary constraint but makes the problem feasible, enabling us to iteratively employ the SCA technique to solve the resulting optimization problem. To proceed, we rewrite the problem $(\P1)$ approximately as
\begin{subequations}\label{usrsch_subprob_cvx}
\begin{align}
(\P1.1):& \stackrel{}{\underset{\pmb{\alpha}, \eta}{\mathrm{maximize}}~~~\psi - \mu \eta}   \nonumber\\
\mathrm{s.t.}&\quad \eqref{P1_cond1}, \eqref{P1_cond2}, \eqref{P1_cond4}\\
&\hspace{-10mm} \frac{1}{N}\sum^{N}_{n=1} A_{k,n}\alpha_k[n] \geq \psi,~\forall k\in\K\\
&\hspace{-10mm}   \sum^{N}_{n=1}\sum^{K}_{k=1}\left[(1- 2\alpha^{lo}_{k}[n]) \alpha_k[n] +(\alpha^{lo}_{k}[n])^2\right] \leq  \eta \label{39b},
\end{align}
\end{subequations}
where $\eta$ is a non-negative slack variable, and $\mu$ is the given penalty parameter. It is worth stressing that feasible set of problem $(\P1.1)$ is larger than that of $(\P1)$; therefore, by choosing a small initial value for $\mu$ we can make problem $(\P1.1)$ feasible, and then by gradually increasing $\mu$, or in some sense, forcing $\eta$ to approach zero via maximizing the given objective function, we can reach the optimal solution of $(\P1.1)$ which can serve as the optimal solution for $(\P1)$ once $\eta=0$ is attained. Note that problem $(\P1.1)$ is indeed a linear programming (LP) that can be readily handled in polynomial time order via a convex optimization toolbox, e.g., CVX \cite{CVXResearch2012}.

\subsection{Sub-problem II: Joint Power Allocation Optimization}\label{P3_solution}
In this subsection, we tackle the optimization of network power allocations. To this end, the corresponding sub-problem for jointly optimizing AP's transmission power $\Pow_a$ as well as maximum AN transmissions of UCJ $\Pow_j$ can be given as
\begin{subequations}
\begin{align}
(\P2):& \stackrel{}{\underset{\Pow_a, \hat{\Pow}_j}{\mathrm{maximize}}~~\underset{k\in\K}{\min}~\frac{1}{N}\sum^{N}_{n=1}A_{k,n}\ln\left(1+\frac{B_{k,n} p_a[n]}{C_{k,n} \hat{p}_j[n] + 1}\right)}   \nonumber\\
\mathrm{s.t.}&\quad \eqref{pow_conds}, \label{P2_cond1}\\
& \hspace{-7mm}\sum^{K}_{k=1} \alpha_k[n] \min_{m \in \W} \left(1 \hspace{-1mm}-\hspace{-1mm} D_{n,k,m} \frac{p_a[n]}{p_j[n]}\right) \hspace{-1mm}\geq\hspace{-1mm} 1\hspace{-1mm}-\hspace{-1mm}\varepsilon,~\forall n\in\N  \label{P2_cond2}
\end{align}
\end{subequations}

where, for $~\forall n \in \N, \forall k \in \K, \forall m  \in \W$,
\[
A_{k,n} \hspace{-1mm}=\hspace{-1mm} \frac{\mathrm{W}\alpha_k[n]}{N\ln(2)(P_{f,r}[n] + P_{f,j}[n])},
~B_{k,n} \hspace{-1mm}= \hspace{-1mm}\frac{g_{ark}[n]}{\sigma^2_k[n]},\]
\[C_{k,n}\hspace{-1mm} = \hspace{-1mm}\frac{g_{jk}[n]}{2\sigma^2_k[n]},~D_{n,k,m} \hspace{-1mm} =\hspace{-1mm} \frac{g_{arm}[n]}{g_{jm}[n]}\]
Sub-problem $(\P2)$ is nonconvex due to nonconvex objective function and constraint \eqref{P2_cond2}. First, we tackle the nonconvex constraint \eqref{P2_cond2} by introducing non-negative slack variables $\mathbf{S}=\{s_k[n], \forall k \in \K, \forall n \in \N\}$, we rewrite $(\P2)$ as
\begin{align}
(\P2.1):& \stackrel{}{\underset{\Pow_a, \hat{\Pow}_j, \mathbf{S}}{\mathrm{maximize}}~~\underset{k\in\K}{\min}~\sum^{N}_{n=1}A_{k,n}\ln\left(1+\frac{B_{k,n} p_a[n]}{C_{k,n} \hat{p}_j[n] + 1}\right)}   \nonumber\\
\mathrm{s.t.}&\quad\eqref{P2_cond1}, \sum^{K}_{k=1} \alpha_k[n] s_k[n] \geq 1-\varepsilon,~ \forall n\in\N \label{P21_cond1}\\
&
\begin{aligned}
\ln(\hat{p}_j[n]) &+ \ln(1-s_k[n]) \geq \ln(D_{n,k,m} p_a[n]),\label{P21_cond2}\\
&\forall n\in \N, k\in\K, m\in\W
\end{aligned}
\end{align}
The obtained formulation $(\P2.1)$ is in good shape but still nonconvex due to the nonconvex objective function and \eqref{P21_cond2}. Since the summation terms of the objective function of problem $(\P2.1)$ are in the form of \textit{concave-minus-concave} owing to the fact that the logarithm of any non-negative affine function is concave, we substitute the objective function with a non-negative slack variable $\eta$, then replacing the sum of logarithmic functions with their corresponding global concave lower bound, as well as replacing \eqref{P21_cond2} with the approximate convex constraint using the  first-order restrictive law of Taylor approximation at the given local point $\hat{\Pow}^{lo}_j=\{\hat{p}^{lo}_j[n],~\forall n\in\N\}$, we can rewrite $(\P2.1)$ as the following convex reformulation 
\begin{align}\label{jointpow_cvx}
(\P2.2):& \stackrel{}{\underset{\Pow_a, \hat{\Pow}_j, \mathbf{S}, \eta}{\mathrm{maximize}}~~\eta}  \nonumber\\
\mathrm{s.t.}&\quad 
\begin{aligned}
\frac{1}{N}\sum^{N}_{n=1}&A_{k,n}\bigg[\ln\left(1+B_{k,n} p_a[n] + C_{k,n} \hat{p}_j[n]\right)\\
&\hspace{-10mm}-f_{2}(\hat{p}_j[n],\hat{p}^{lo}_j[n]) \bigg] \geq \eta,~\forall k\in \K\
\end{aligned}\\
&\begin{aligned}
\eqref{P21_cond1},~ &\ln(\hat{p}_j[n]) + \ln(1-s_k[n])\geq g_2(p_a[n],p^{lo}_a[n]),\\
&~\forall n\in \N, k\in\K, m\in\W 
\end{aligned}
\end{align}
where 
\begin{align}
    f_2(\hat{p}_j[n],\hat{p}^{lo}_j[n]) &= \ln(1+C_{k,n} \hat{p}^{lo}_j[n]) \nonumber\\
    &\hspace{-10mm}+ \frac{C_{k,n}(\hat{p}_j[n] - \hat{p}^{lo}_j[n])}{1+C_{k,n}\hat{p}^{lo}_j[n] },~\forall n\in \N, k\in\K
\end{align} 
\begin{align}
  g_2(p_a[n],p^{lo}_a[n]) &= \ln(D_{n,k,m} p^{lo}_a[n]) \nonumber\\
  &\hspace{-10mm}+ \frac{p_a[n]-p^{lo}_a[n]}{p^{lo}_a[n]},~\forall n\in \N, k\in\K, m\in\W
\end{align}

Since problem $(\P2.2)$ is a conic convex optimization problem, it can be efficiently solved using CVX. 

Note that, with more relaxations, a joint design of user scheduling and network power allocation can be presented at the cost of relatively higher complexity  than individual design as given in Appendix \ref{Appendix A}.

\subsection{Sub-problem III: Beamforming optimization}\label{P2_solution}

To optimize beamforming matrices, we can write the corresponding subproblem to optimize $\pmb{\Phi}$, given the other variables fixed, as
\begin{subequations}\label{beamforming_subprob_nth}
\begin{align}
\scriptsize
(\P3)\hspace{-1mm}:& \stackrel{} {\underset{\pmb{\Phi}}
{\mathrm{maximize}}~\underset{k\in\K}{\min}~\frac{\sum_{n}\hspace{-1mm} A_{k,n}\hspace{-1mm}\ln\hspace{-1mm}\left(1\hspace{-1mm} +\hspace{-1mm} B_{k,n}\|\e^\dagger_{rk}[n]\Phi[n]\e_{ar}[n]\|^2 \right)}{N}} \nonumber\\
\mathrm{s.t.}&\quad
\begin{aligned}
\sum^K_{k=1}\alpha_k[n]&\underset{m\in\W}{\min}~\left[1-C_{m,k,n} \|\e^\dagger_{rm}[n]\Phi[n]\e_{ar}[n]\|^2 \right]\\
&\geq 1 - \varepsilon,~\forall n\in\N
\end{aligned}\\
&0 \leq \rho_n[l] \leq 1,~\forall l\in\L,~n\in\N \\
&0 \leq \phi_n[l] \leq 2\pi,~\forall l\in\L,~n\in\N 
\end{align}
\end{subequations}
where  $\pmb{\rho}=\{\rho_n[l],~\forall l\in\L,~n\in\N\}$ and $\pmb{\phi}=\{\phi_n[l],~\forall l\in\L,~n\in\N\}$, and for $\forall n\in\N,~k\in\K,~m\in\W$
\[
A_{k,n} = \frac{\alpha_k[n]\mathrm{W}}{P_{f,r}[n] + P_{f,j}[n]},~~~B_{k,n} = \frac{p_a[n]  \|\tilde{h}_{ark}[n]\|^2}{\frac{1}{2}\hat{p}_j[n]g_{jk} + \sigma^2_k[n]},
\]
\[
C_{m,k,n} = \frac{p_a[n]\|\tilde{h}_{arm}[n]\|^2}{\hat{p}_j[n]g_{jm}},
\]

Now, exploiting the change of variable technique and defining the slack vector $\u[n]=[u_1[n], u_2[n], \cdots, u_L[n]]^\mathrm{T}$ with $u_l[n]=\rho_l[n]\exp(-j\phi_l[n]),~\forall l\in\L,~n\in\N$, and taking the auxiliary column vectors $\mathbf{a}_k[n]=\mydiag\left(\e^\dagger_{rk}[n] \right)\e_{ar}[n],~\forall k\in\K,n\in\N$, $\mathbf{b}_m[n]=\mydiag\left(\e^\dagger_{rm}[n] \right)\e_{ar}[n]$, we can reformulate $(\P3)$, introducing the slack variables $\pmb{\eta} = \{\eta_k[n],~\forall n\in\N,~k\in\K\}$ as 
\begin{subequations}\label{P3_subproblem1}
\begin{align}
(\P3.1):& \stackrel{}{\underset{\pmb{\eta},~ \u}{\mathrm{maximize}}~\underset{k\in\K}{\min}~\frac{1}{N}\sum^{N}_{n=1}A_{k,n}\ln\left(1+B_{k,n} \eta_k[n]\right) }\nonumber\\
\mathrm{s.t.}&\quad \u^\dagger[n]\mathbf{A}_k[n]\u[n] \geq \eta_k[n],~\forall n\in\N,~k\in \K\label{P31_cond1}\\
&\begin{aligned}\\
\sum^K_{k=1}\alpha_k[n]&\underset{m\in\W}{\min}~\left[1-C_{m,k,n} \u^\dagger[n]\mathbf{B}_m[n]\u[n] \right]\\
&\geq 1 - \varepsilon,~\forall n\in\N
\end{aligned}\\
& \|\u[n]\| \leq 1,~\forall n\in\N\label{P31_cond_magnitude}
\end{align}
\end{subequations}
where $\mathbf{A}_k[n] = \mathbf{a}_k[n]\mathbf{a}^\dagger_k[n]$, $\mathbf{B}_m[n] = C_{m,k,n}\mathbf{b}_m[n]\mathbf{b}^\dagger_m[n]$. 

\renewcommand{\W}{\ensuremath{\mathbf{W}}}
Note that the objective function is in the form of the non-negative combination of logarithms, each of which is non-decreasing w.r.t $\eta_k[n]$. Thus the inequality \eqref{P31_cond1} must be met with equality at the optimal point; otherwise, the objective value can be further improved, violating the optimality. Problem $(\P3.1)$ is a quadratically constrained nonlinear program (\textbf{QCNLP}) which is NP-hard, so is its relaxed version: quadratically constrained quadratic programming (\textbf{QCQP}). Therefore, to facilitate the development of problem, we define $N$ matrices of size $L$-by-$L$, i.e., $\mathbb{W}\treq\{\W[n]\treq\{\u[n] \u^\dagger[n],~\forall n\in\N\}$, which ensures $\mathbb{W}$ to be rank one Hermitian positive semidefinite (PSD) matrix set. Then we recast $(\P3.1)$ as  
\begin{subequations}\label{P3_subproblem2}
\begin{align}
(\P3.2):& \stackrel{}{\underset{\pmb{\eta},~\mathbb{W}}{\mathrm{maximize}}~\underset{k\in\K}{\min}~\frac{1}{N}\sum^{N}_{n=1}A_{k,n}\ln\left(1+B_{k,n} \eta_k[n]\right) }\nonumber\\
\mathrm{s.t.}&\quad \mytrace(\mathbf{A}_k[n]\W[n]) \geq \eta_k[n],~\forall n\in\N,~k\in \K\label{P32_cond1}\\
&\begin{aligned}\label{P32_cond2}\\
\sum^K_{k=1}\alpha_k[n]&\underset{m\in\W}{\min}~\left[1-C_{m,k,n}  \mytrace(\mathbf{B}_m[n]\W[n])  \right]\\
&\geq 1 - \varepsilon,~\forall n\in\N
\end{aligned}\\
&\qquad\quad \W_{l,l}[n] \leq 1,~\forall l\in\L,~n\in\N\label{magnitude_cond}\\
&\qquad\quad\W[n]  \succeq 0,~\forall n\in\N \label{LMI}\\
&\qquad\quad\myrank(\W[n])=1,~\forall n\in\N \label{rank-one}
\end{align}
\end{subequations}
where $\W_{l,l}[n]$ in \eqref{magnitude_cond} refer to the elements alongside the main diagonal of the matrix of $\W[n]$ and is the reformulation of the magnitude constraint \eqref{P31_cond_magnitude}, \eqref{LMI} is the linear matrix inequality (LMI) indicating that $\W[n]\in\mathbb{S}^+$  and this constraint is convex w.r.t the optimization variables. It is worth pointing out that by such reformulation, we have converted nonconvex objective constraints in \eqref{P3_subproblem1} into convex equivalent constraints, but unfortunately $(\P3.2)$ is as hard as $(\P3.1)$ to solve due to the newly introduced nonconvex rank constraint \eqref{rank-one}. Nonetheless, dropping such nonconvex constraint, we can obtain the semi-definite relaxation (SDR) reformulation of \eqref{P3_subproblem2} as 
\begin{subequations}\label{P3_subproblem21}
\begin{align}
(\P3.2.1):& \stackrel{}{\underset{\pmb{\eta},~\mathbb{W}}{\mathrm{maximize}}~\underset{k\in\K}{\min}~\frac{1}{N}\sum^{N}_{n=1}A_{k,n}\ln\left(1+B_{k,n} \eta_k[n]\right) }\nonumber\\
\mathrm{s.t.}&\quad\eqref{P32_cond1}, \eqref{P32_cond2}, \eqref{magnitude_cond}, \eqref{LMI} \label{prob321_cond1}
\end{align}
\end{subequations}
Note that problem $(\P3.2.1)$  is a convex semi-definite programming (\textbf{SDP}), which can be solved in a numerically reliable and efficient manner by applying the interior-point method \cite{SDP2010Quan}. Though \textbf{SDR} is a computationally efficient approximation approach, problem $(\P3.2.1)$ does not necessarily generate rank-one matrices as its optimal solution $\mathbb{W}^\star=\{\W^\star[n],~\forall n\in\N\}$.  Indeed, if $\W^\star[n]$ for $\forall n\in\N$ are of rank one, then we can readily write $\W^\star[n] = \u^\star[n]{\u^\star}^\dagger[n]$ ($\u^\star[n],~\forall n\in\N$ being the eigenvector corresponding the only non-zero eigenvalue of $\W^\star[n]$) extracting not only feasible but also optimal solution $\u^\star[n]$ to problem $(\P3.1)$ which is consequently the optimal solution to problem $(\P3)$. Otherwise, if the rank of $\W^\star[n]$ is larger than one, the \textbf{SDR} method can only generate a tight lower-bound on the optimal objective value of $(\P3.2)$.
To this end, we proceed to tackle this issue by employing an iterative rank minimization approach, namely \textit{rank-penalty SCA} (\textbf{RP-SCA}), as discussed below.

The essential notion of \textbf{RP-SCA} technique lies in the fact that every rank-one matrix has only one non-zero eigenvalue. Hence, instead of making constraints on the rank, we focus on the $L-1$ smallest eigenvalues of $\W[n],\forall n$ forcing them to be all zero via an efficient iterative approach, as discussed in the lemma below.                

\renewcommand{\W}{\ensuremath{\mathbf{W}}}
\renewcommand{\V}{\ensuremath{\mathbf{V}}}
\renewcommand{\C}{\ensuremath{\mathbb{C}}}
\renewcommand{\S}{\ensuremath{\mathbb{S}}}

\begin{lemma}\label{lemma1:rank1refomulaiton}
Let $\W\in\S^+$ be a nonzero PSD matrix, then 
\[ \myrank(\W)=1 \iff \psi I - \V^\dagger\W\V \succeq 0,\]
where $\psi =0$, $\I$ is the identity matrix of size $(L-1)$, and $\V \in \C^{{L}\times{L-1}}$
are the eigenvectors corresponding to the $(L-1)$ smallest eigenvalues of $\W$.

\begin{proof}
Since $\W$ is PSD, it is Hermitian with nonnegative eigenvalues (please see Appendix \ref{Appendix B}).
We assume that the nonnegative eigenvalues of $\W$ are sorted in ascending order as $[\lambda_1, \lambda_2, \cdots, \lambda_L]$ with the $(L-1)$ corresponding normalized eigenvectors $\V=[\v_1, \v_2, \cdots, \v_{L-1}]$. Since the Rayleigh quotient of an eigenvector equals with its associated eigenvalue, i.e., $\v^\dagger_l\W\v_l=\lambda_l,~\forall l$ and eigenvectors of distinct eigenvalues for any Hermitian matrix are orthogonal, thus  $\psi I - \V^\dagger\W\V = \mydiag(\psi-\lambda_1, \psi-\lambda_2, \cdots, \psi-\lambda_{L-1})$ is a diagonal matrix. Therefore, $\W$ is rank one iff all the diagonal elements of $\psi I - \V^\dagger\W\V$ are all zero given $\psi=0$. 

\end{proof}

\end{lemma}

Using Lemma \ref{lemma1:rank1refomulaiton}, we replace the rank constraint \eqref{rank-one} in $(\P2.2)$ with the corresponding semi-definite constraint and introduce the slack variables $\nu$ and $\pmb{\psi}=\{\psi[n],~n\in\N\}$, then rewrite ($\P3.2$) as

\begin{subequations}\label{P3_subproblem3}
\begin{align}
(\P3.3):& \stackrel{}{\underset{\nu,\pmb{\psi},\pmb{\eta},\mathbb{W}}{\mathrm{maximize}}~\nu - \sum^N_{n=1}\mu \psi[n] }\nonumber\\
\hspace{-5mm}\mathrm{s.t.}&~ \frac{1}{N}\sum^{N}_{n=1}A_{k,n}\ln\left(1+B_{k,n} \eta_k[n]\right) \geq \nu,~\forall k\in\K\\
& \psi[n] I - \V[n]^\dagger\W[n]\V[n] \succeq 0, ~\forall n\in \N\\
&\pmb{\psi} \geq \mathbf{0},~\eqref{prob321_cond1}
\end{align}
\end{subequations}

\renewcommand{\C}{\ensuremath{\mathbb{C}}}
\newcommand{\R}{\ensuremath{\mathbb{R}}}
\renewcommand{\v}{\ensuremath{\mathbf{v}}}

 It is worth mentioning that directly applying rank-one reformulation in Lemma \ref{lemma1:rank1refomulaiton} likely results in an infeasible problem. Therefore, by introducing a slack variable and bringing the rank-penalty to the objective function with weighting factor $\mu$, we attempt to iteratively maximize the original objective value while minimizing the rank penalty term. Since $\W[n]\in \mathbb{S}^+,~\forall n$ then $\V^\dagger\W\V \in \mathbb{S}^+$ ($(\V^\dagger\W\V)^\dagger=\V^\dagger\W\V$). Thus, $\pmb{\psi}$ needs to be sufficiently large and nonnegative to make the problem feasible. Therefore,  less focus initially is on the rank-one requirement, but rank one is satisfied when $\psi=0$. Once the rank-one PSD matrix $\W^\star$ is obtained, the optimal solution of ($\P3.1$) $\u^\star$ can be computed via $\u^\star=\sqrt{\lambda_L}\v_L$, where $\lambda_L\in\R^+$ is the largest eigenvalue of $\W^\star$ and $\v_L \in \C^L$ is the corresponding normalized eigenvector. We summarize the proposed beamforming maximization in Algorithm \ref{algo:beamforming}.

\renewcommand{\W}{\ensuremath{\mathbf{W}}}
\renewcommand{\V}{\ensuremath{\mathbf{V}}}
\renewcommand{\A}{\ensuremath{\mathbf{A}}}
\newcommand{\B}{\ensuremath{\mathbf{B}}}

\begin{figure}[!t]
 \removelatexerror
  \begin{algorithm}[H]
   \SetKwBlock{DoParallel}{}{}
  \caption{ Iterative \textbf{RM-SCA} algorithm for UIRS's beamforming maximization}\label{algo:beamforming}
  1:~\textbf{Initialize}~ set of matrices $\{\A_k[n], \B_m[n], \forall n,k,m \}$, small value $\mu$, parameters $\mu_{max}$ and $\varrho$\\
    2: Solve $(\P3.2.1)$ using \eqref{P3_subproblem21} to obtain $\W^{(0)}_n,~\forall n$ then extract $\V^{(0)}_n,~\forall n$, and set iteration index $i=0$, and small value $\mu$\\
    3:~\textbf{Repeat:} 
    \DoParallel{
    3.1:~$i \gets i + 1$\\
    3.2:~Solve $(\P3.3)$ using \eqref{P3_subproblem3} to find $\mathbb{W}^{(i)}$ and $\pmb{\psi}^{(i)}$\\
    3.3:~Obtain $\V^{(i)}_n$ from $\W^{(i)}_n$\\
    3.4:~$\mu \gets \min\{\mu_{max},\varrho \mu\}$}
    4:~\textbf{Until} Convergence\\
    5:~\textbf{Extract} $\u[n]$ from $\W^{(i)}[n],~\forall n$;\\
    6:~\textbf{Return:} $\{\pmb{\rho}^\star[n] = \|\u[n]\|, \pmb{\phi}^\star[n]=\measuredangle\u[n] \}^{N}_{n=1}$
  \end{algorithm}
\end{figure}

\renewcommand{\W}{\ensuremath{\mathcal{W}}}
Now we present a low-complex alternative design to balance mAEE performance and computational complexity. Note that with the known scheduled UE per time slot, the fixed transmit powers and UAVs' trajectories, we can maximize the quality of the received signal at Bob via coherently combing signals from different paths, meanwhile degrading the leakage to Willies. Given UE $k$ is scheduled per time slot, we can rewrite \eqref{ap-uirs-ue_ch_gain}
as
\begin{align}\label{cascadedchannelgain}
  &h_{ark}[n] = \tilde{h}_{ark}[n] \nonumber\\
  &\times\sum^L_{l=1} \rho_l[n]\mathrm{e}^{j(\phi_l[n] + \beta^k_l[n] - \theta^k_l[n])},~\forall k\in \K,~n\in \N
\end{align}
Note that improving received signal energy at Bob via improving the channel gain \eqref{cascadedchannelgain} leads to mAEE improvement according to \eqref{rk_lb} and \eqref{opt_prob}. Thus, by properly aligning the phases of different paths, we can improve the channel quality by setting $\phi_l[n] + \beta^k_l[n] - \theta^k_l[n]=\omega,~\rho^\star_l[n]=1.~\forall l\in\L,~n\in\N$,where $0 \leq \omega \leq 2\pi$, the optimized IRS beamforming matrices can be obtained as 
\begin{align}\label{IRSbeamforming_opt}
    \Phi^\star[n] = \mydiag \left(\{\rho^\star_l[n]\mathrm{e}^{j\phi^\star_l[n]}\}^L_{l=1}\right),~ \forall l\in\L,~n\in\N
\end{align} 
with
\[
   \hspace{-3mm}\phi^\star_l[n] = \omega + \theta^k_l[n] - \beta^k_l[n],~~~\rho^\star_l[n] = 1,
 \]
   
Here, the magnitude of the IRS reflected beam needs to be set unity, i.e., $\rho_l[n]=1,~\forall l\in\L,~n\in\N$; otherwise  $\rho_l[n]$ can be increased, resulting in further improvement of $h_{ark}[n]$ and then our objective function in terms of the mAEE. This demonstrates the sub-optimality of the low-complex design. Consequently, \eqref{cascadedchannelgain} is maximized and accordingly $\bar{R}^{lb}_k[n]$ given by \eqref{rk_lb} can be simplified as
\renewcommand{\C}{\ensuremath{\mathrm{C}}}

\begin{align}
  \bar{R}^{lb}_k[n] \hspace{-1mm}=\hspace{-1mm}\mathrm{W}\log_2\hspace{-1mm}\left[\hspace{-1mm}1\hspace{-1mm}+\hspace{-1mm}\frac{p_a[n]\left(\frac{L^2h_{r0}[n]\mathrm{e}^{-{\kappa(\|\q_r[n] - \q_a \| + \|\q_r[n] - \q_k \|)}} }{\|\q_r[n] - \q_a \|^{{\rho}}\|\q_r[n] - \q_k \|^{{\rho}}}\right)}{\frac{1}{2}\hat{p}_j[n]\frac{h_{j0}[n]\mathrm{e}^{-{\kappa \|\q_j[n] - \q_k \|}}}{\|\q_j[n] - \q_k \|^{\rho}} +1}\hspace{-1mm}\right],
\end{align}
where $h_{r0}[n] \treq \frac{\lambda^2_c}{64\pi^3 \sigma^2_k[n]}$ and $h_{j0}[n] \treq \frac{\lambda^2_c}{16\pi^2 \sigma^2_k[n]}$. Furthermore, following the above beamforming design, the minimum detection error rate given by \eqref{zeta_min} can be rewritten as
\begin{align}\label{zeta_min_recast}
    \zeta^\star_{m,k}[n] &= 1 - \frac{p_a[n]\frac{\tilde{L}[n]h_{r0}\exp\left(-\kappa(\|\q_r[n] - \q_a \| + \|\q_r[n] - \q_m \|)\right)} {\|\q_r[n] - \q_a \|^{{\rho}}\|\q_r[n] - \q_m \|^{{\rho}}}}
    {\hat{p}_j[n]\frac{h_{j0}\exp\left(-\kappa \|\q_j[n] - \q_m \|\right)}{\|\q_j[n] - \q_m \|^{\rho}}},\nonumber\\
    &\qquad\forall n\in\N,~k\in\K,~m\in\W
\end{align}
where 
\[\tilde{L}[n] \treq \bigg\|\sum^L_{n=1}\exp\left(j(\phi^\star_l[n] - \theta^m_l[n] + \beta^m_l[n] )\right)\bigg\|^2\]

Since \eqref{zeta_min_recast} is too complicated, we consider using its restrictive conservative lower-bound, leading to the reformulation of \eqref{covert_cond} as
 \begin{align}
\widetilde{\C7}:~\sum^{K}_{k=1} \alpha_k[n] &\min_{m \in \W} \bigg[1 - \frac{p_a[n]}{\hat{p}_j[n]}\nonumber\\
&\times\frac{\frac{L^2h_{r0}\exp\left(-\kappa(\|\q_r[n] - \q_a \| + \|\q_r[n] - \q_m \|)\right)} {\|\q_r[n] - \q_a \|^{{\rho}}\|\q_r[n] - \q_m \|^{{\rho}}}}
    {\frac{h_{j0}\exp\left(-\kappa \|\q_j[n] - \q_m \|\right)}{\|\q_j[n] - \q_m \|^{\rho}}}\bigg] \nonumber\\
    &\geq 1-\varepsilon,~ \forall n\in\N,~k\in\K \label{covert_cond_recasted1}
 \end{align}
In the sequel, we utilise above reformulations for trajectory and velocity optimization.
\subsection{Sub-problem IV: Joint UCJ's Trajectory and Velocity Optimization}
We focus on jointly optimizing the UCJ's trajectory $\q_j$ and velocity $\v_j$ while keeping the transmit powers, UIRS's beamforming matrix, and  user scheduling $\pmb{\alpha}$ fixed. The corresponding sub-problem is given as problem $(\P4)$ (shown on top of the next page), 
\begin{figure*}[t]
\begin{subequations}
\begin{align}
(\P4):& \stackrel{}{\underset{\q_j,\v_j}{\mathrm{maximize}}~\underset{k\in\K}{\min}~\frac{1}{N}\sum^N_{n=1}\frac{A_{k,n}\ln\left(1+\frac{B_{k,n}}{C_{k,n}\frac{\exp(-\kappa\|q_j[n] - \q_k\|)}{\|\q_j[n]-\q_k\|^\rho} + 1}\right)}{P_o\left(1+c_0\|\v_j[n]\|^2\right) + c_1 \|\v_j[n]\|^3  +  P_i\left(\sqrt{1+c^2_2\|\v_j[n]\|^4} -c_2\v_j[n]\|^2\right)^{\frac{1}{2}} + P_{f,r}[n]}} \nonumber\\
&\mathrm{s.t.}~~~~~\eqref{ucj_flight_conds},~\sum^K_{k=1}\alpha_k[n]\underset{m\in\W}{\min}~\left\{1-D_{m,k,n}\|\q_j[n]-\q_m\|^\rho \exp(\kappa \|\q_j[n] - \q_m\|)\right\} \geq 1-\varepsilon,~\forall n\in\N, m\in\W 
\label{p4_cond2}
\end{align}
\end{subequations}
\noindent\rule{\textwidth}{.5pt}
\end{figure*}
where
\[A_{k,n} \hspace{-1mm}=\hspace{-1mm} \frac{\alpha_k[n]W}{\ln(2)},~B_{k,n}\hspace{-1mm}=\hspace{-1mm}\frac{p_a[n] g_{ark}[n]}{\sigma^2_k},~C_{k,n} \hspace{-1mm}=\hspace{-1mm}\left(\frac{\lambda_c}{4\pi}\right)^2  \frac{\hat{p}_j[n]}{2\sigma^2_k}\]
\[D_{m,k,n}= \frac{p_a[n] g_{arm}[n]}{p_j[n] \left(\frac{\lambda_c}{4\pi}\right)^2},~\forall n\in\N, k\in \K, m\in \W\]
Problem $(\P4)$ is nonconvex  being in the form of minimum of sum-of-fractional-nonlinear-programming. To facilitate the development of optimization framework, we first mention the following lemma.

\newcommand{\x}{\ensuremath{\mathbf{x}}}
\newcommand{\y}{\ensuremath{\mathbf{y}}}
\newcommand{\X}{\ensuremath{\mathcal{X}}}
\newcommand{\Y}{\ensuremath{\mathcal{Y}}}

\begin{lemma}\label{lemma2:sum-of-ratios}
Given $N$ pairs of non-negative functions $f_n(\x)$ and positive functions $g_n(\x)$, as well as having a nonempty constraint set $\mathcal{X}$, the general average-of-ratios fractional programming (FP) maximization can be represented in the form of
\begin{align}\label{FP:general}
&\underset{\x}{\mathrm{maximize}}~~\frac{1}{N}\sum^N_{n=1}\frac{f_n(\x)}{g_n(\x)} \nonumber\\
&~\mathrm{s.t.}~~~\x \in \X
\end{align}
The above problem is in general nonconvex, thereby too difficult to solve. However, if we could obtain a tight concave lower-bounds on $f_n(\x),~\forall n$, represented by $f^{lb}_n(\y),~\forall n$, and the tight convex upper-bounds of $g_n(\x),~\forall n$, denoted as $g^{up}_n(\x),~\forall n$, such that $f_n(\x) \geq f^{lb}_n(\y)~\text{and}~g_n(\x) \leq g^{up}_n(\y), \forall n$, then by applying the Dinkelbach-based quadratic transformation introduced in \cite{FP2018Shen}, \eqref{FP:general} can be  approximately reformulated  as a convex optimization problem given by
\begin{align}\label{FP:cvx}
&\underset{\y, \pmb{\gamma}}{\mathrm{maximize}}~~\frac{1}{N}\sum^N_{n=1}\left(2\gamma^2_n\sqrt{f^{lb}_n(\y)} - \gamma^2_n g^{up}_n(\y)\right) \nonumber\\
&~\mathrm{s.t.}~~~\y \in \Y\subseteq\X
\end{align}
where $\pmb{\gamma}=\{\gamma_{n}\}^N_{n=1}$ are auxiliary variables whose optimal values with fixed $\y$ can be obtained as $\pmb{\gamma}^\star = \frac{\sqrt{f^{lb}_n(\y)}}{g^{up}_n(\y)}$. It is also worth stressing that the optimal value of problem \eqref{FP:general} is no less than that of \eqref{FP:cvx}. Since problem \eqref{FP:cvx} is a convex optimization problem; thus, it can be solved efficiently by alternatively updating $\pmb{\gamma}$ and $\y$ starting from a feasible point. Consequently, having solved \eqref{FP:cvx}, we can achieve an approximate solution of the original problem given by \eqref{FP:general} but with guaranteed convergence.
\end{lemma}

\newcommand{\w}{\ensuremath{\mathbf{w}}}
\renewcommand{\t}{\ensuremath{\mathbf{t}}}

Using Lemma \ref{lemma2:sum-of-ratios}, we first convert each ratio of summation terms in the objective function of $(\P4)$ to be in the form of non-negative concave over positive convex functions. To this end, we take non-negative slack variables  $\w=\{w_k[n],~\forall n\in\N, k\in\K\}$, $\t = \{t[n],~\forall n\in \N\}$ such that
\[w_k[n] = \frac{\exp(-\kappa\|q_j[n] - \q_k[n]\|)}{\|\q_j[n]-\q_k[n]\|^\rho},~\forall n\in\N, k\in\K\] 
and \[t[n] = \left(\sqrt{1+c^2_2\|\v_j[n]\|^4} -c_2\|\v_j[n]\|^2\right)^{\frac{1}{2}},~\forall n\in\N\]
Since 
\[f_4(w[n])\treq A_{k,n}\ln(1+\frac{B_{k,n}}{C_{k,n} w_k[n]+1})\] is convex w.r.t $w_k[n]$, according to \cite[Lemma 1]{mamaghani2020_tvt}, we can apply the first order restrictive approximation to obtain a global concave lower-bound at the given local point $\w^{lo}=\{w^{lo}_k[n], \forall n\in\N, k\in\K\}$ as 
\begin{align}
f_4(w_k[n]) &\geq f_4(w^{lo}_k[n]) \nonumber\\
&\hspace{-15mm}- \frac{A_{k,n}B_{k,n}C_{k,n}(w_k[n]\hspace{-1mm}-\hspace{-1mm}w^{lo}_k[n])}{(C_{k,n} w^{lo}_k[n]\hspace{-1mm}+\hspace{-1mm}1)(C_{k,n} w^{lo}_k[n]\hspace{-1mm}+\hspace{-1mm}B_{k,n}\hspace{-1mm}+\hspace{-1mm}1)}\hspace{-1mm}\treq\hspace{-1mm} f^{lb}_4(w_k[n]),
\end{align}
We reformulate problem ($\P4$), taking the new non-negative slack variables $\v=\{v_k[n], \forall n\in\N, k\in\K\}$, as
\begin{subequations}\label{subprob4.1}
\begin{align}
(\mathrm{P4.1}):& \stackrel{}{\underset{\q_j,\v_j, \w, \v, \t}{\mathrm{maximize}}~~~
\begin{aligned}
\underset{k\in\K}{\min}~\frac{1}{N}\sum^N_{n=1}\bigg(2\gamma_{k,n}&\sqrt{f^{lb}_4(w_k[n])}  \nonumber\\
&\hspace{-5mm}- \gamma^2_{k,n}  g^{up}_4(\v_j[n],t[n])\bigg)
\end{aligned}}\nonumber\\
\mathrm{s.t.}&\quad
\begin{aligned}\label{p41_cond1}
\eqref{p4_cond2},&~\ln(w_k[n]) + \rho\ln(v[n])\\
&+\kappa v[n] \geq 0,~\forall n\in\N, k\in\K 
\end{aligned}\\
&\quad v_k[n] \leq \|\q_j[n] - \q_k \|,~\forall n\in\N, k\in\K \label{p41_cst2}\\
&\quad t^2[n] + c_2\|\v_j[n]\|^2 \geq \frac{1}{t^2[n]},~\forall n\in \N \label{p41_cst3}\\
&\quad \|\q_r[n] - \q_j[n]\| \geq D_s,~\forall n\in\N\label{p41_cst5}
\end{align}
\end{subequations}
where $g^{up}_4(\v_j[n],t[n])$ is a convex function given by
\begin{align}
   g^{up}_4(\v_j[n],t[n]) &=  \Big(P_o\left(1+c_0\|\v_j[n]\|^2\right) + c_1 \|\v_j[n]\|^3 \nonumber\\
   &+  P_it[n] + P_{f,r}[n]\Big),~\forall n\in\N
\end{align}
Note that inequality constraints \eqref{p41_cond1}, \eqref{p41_cst2}, and \eqref{p41_cst3} must be met with equality at the optimal point on the grounds that the value of the objective function can be otherwise increased  without violating any constraints. Note that the convexity of constraint \eqref{p41_cond1} follows from the fact that $\log(x), x>0$ is a concave function of $x$ and the sum operator preserves the convexity. Further, knowing that $h_1(x)=x^\rho\exp(\kappa x), x \geq 0$ with $\rho \geq 2$ is a convex function whose extended value extension is non-decreasing w.r.t $x$, plus $h_2(\q_j[n])=\|\q_j[n]-\q_m\|$ is convex as any norm of affine function is convex, one can conclude that the composition function $h_1 \circ h_2 (x)$ is convex \cite{Boyd2004}. With the new reformulation, although $(\mathrm{P4.1})$ is now in comparably good shape, some newly introduced constraints such as \eqref{p41_cst2}, \eqref{p41_cst3} and  \eqref{p41_cst5} are still non-convex. But by substituting their corresponding tight approximate convex constraints via applying the restrictive Taylor approximation at local points $\{\q^{lo}_j[n], t^{lo}[n], \v^{lo}_j[n], \forall n\}$ and introducing a slack variable $\eta_j$, we can reach the convex reformulation of ($\P4.1$) given by
\begin{subequations}\label{subprob_trjUCJ_cvx}
\begin{align}
(\P4.2):& \stackrel{}{\underset{\eta_j, \q_j,\v_j, \w, \v, \t}{\mathrm{maximize}}~~~\eta_j}\nonumber\\
\mathrm{s.t.}&\quad\eqref{p41_cond1},~\begin{aligned}
&\frac{1}{N}\sum^N_{n=1}\bigg(2\gamma_{k,n}\sqrt{f^{lb}_4(w_k[n])}  \\
&\hspace{-15mm}- \gamma^2_{k,n}  g^{up}_4(\v_j[n],t[n])\bigg) \geq \eta_j,~\forall k\in\K\\
\end{aligned}\\
&\begin{aligned}\label{p41_cst2_cvx}
& v^2_k[n] \leq - \|\q_j[n]-\q_k\|^2 \\
&\hspace{-15mm}+ 2(\q_j[n]-\q_k)^\mathrm{T}(\q_j[n]-\q_k),~\forall n\in\N, k\in\K \\
\end{aligned}\\
&\begin{aligned}\label{p41_cst3_cvx}
-(t^{lo}[n])^2 &+2t[n]t^{lo}[n] + c_2\Big(-\|\v^{lo}_j[n]\|^2 \\
&\hspace{-15mm}+ 2(\v^{lo}_j[n])^\mathrm{T}\v_j[n]\Big) \geq  \frac{1}{t^2[n]};~\forall n\in\N \\
\end{aligned}\\
&\begin{aligned}\label{p41_cst5_cvx}
 &2(\q_j[n]-\q_r[n])^\mathrm{T}(\q_j[n]-\q_r[n])\\
 &\hspace{10mm}- \|\q_j[n]-\q_r[n]\|^2, \geq D^2_s,~\forall n\in\N 
\end{aligned}
\end{align}
\end{subequations}
Note that sub-problem $(\P4.2)$,  for a given $\{\gamma_{k,n},~k\in\K,~n\in\N\}$, is equivalent convex approximate representation of sum-of-ratios sub-problem $\mathrm{P4}$ in FP form. Thus, by applying SCA we can solve problem $(\P4.2)$ as summarized in Algorithm \ref{Trj_UCJ_algo}.

\begin{figure}[!t]
 \removelatexerror
 \resizebox{\columnwidth}{!}{%
 \begin{algorithm}[H]
\SetAlgoLined 
\small
\KwResult{$\q^\star_j$, $\v^\star_j$}
\textbf{Initialize} feasible point $(\q^{(0)}_j, \v^{(0)}_j, \w^{(0)}. \t^{(0)})$, set iteration index $s=0$, and define $\pmb{\gamma} = \{\gamma_{k,n},~\forall k, n\}$ such that
$\gamma_{k,n} = \frac{\sqrt{f^{lb}_4(w^{(0)}_k[n]})}{g^{up}_4\left(\v^{(0)}_j[n],t^{(0)}[n]\right)}$, and calculate $\eta^{(0)}_j = \underset{k\in\K}{\min}~\frac{1}{N}\sum^N_{n=1}\frac{{f_4(w^{(0)}_k[n]})}{g^{up}_4\left(\v^{(0)}_j[n],t^{(0)}[n]\right)}$\\
 \While{\textit{true}}{
   $s \gets s + 1$;\\
   Given $\left(\pmb{\gamma}, \q^{(s-1)}_j, \v^{(s-1)}_j, \w^{(s-1)}, \t^{(s-1)}\right)$, solve $(\P4.2)$ using  \eqref{subprob_trjUCJ_cvx}, and obtain $\left(\q^{(s)}_j, \v^{(s)}_j, \w^{(s)}, \t^{(s)}, \eta^{(s)}_j\right)$; \\
  Update $\pmb{\gamma} \gets \frac{\sqrt{f^{lb}_4(\w^{(s)})}}{g^{up}_4(\v^{(s)}_j,\t^{(s)})}$; \\
  \If{$\Big\|\frac{\eta^{(s)}_j - \eta^{(s-1)}_j}{\eta^{(s-1)}_j}\Big\| \leq \epsilon_j$}{
  $(\q^\star_j , \v^\star_j) = (\q^{(s)}_j, \v^{(s)}_j)$;\\
  \textit{break};}
}
 \caption{\label{Trj_UCJ_algo} \small Proposed algorithm to approximately solve subproblem $(\P4)$ alternatively} 
\end{algorithm}}
\end{figure}

\subsection{Sub-problem V: Joint UIRS's Trajectory and Velocity Optimization}
\renewcommand{\x}{\ensuremath{\mathbf{x}}}
\renewcommand{\y}{\ensuremath{\mathbf{y}}}
\renewcommand{\s}{\ensuremath{\mathbf{s}}}
\renewcommand{\t}{\ensuremath{\mathbf{t}}}
\renewcommand{\C}{\ensuremath{\mathrm{C}}}
We now devise the UIRS's both trajectory $\q_r$ and velocity $\v_r$ while keeping the other variables fixed. Thus, the corresponding mAEE optimization sub-problem $(\P5)$ is introduced with constraints in  (\ref{p5_cond2})
\begin{figure*}
\begin{subequations}\label{subprob_trjUIRS}
\begin{align}
(\P5):& \stackrel{}{\underset{\q_r,\v_r}{\mathrm{maximize}}~~~\underset{k\in\K}{\min}~\frac{1}{N}\sum^N_{n=1}\frac{A_{k,n}\ln\left(1+B_{k,n}\frac{\exp\left(-\kappa(\|\q_r[n] - q_a\| + \|\q_r[n] - \q_k\|)\right)}{\|\q_r[n] - q_a\|^\rho\|\q_r[n] - q_k\|^\rho}\right)}{P_o\left(1+c_0\|\v_r[n]\|^2\right) + c_1 \|\v_r[n]\|^3  +  P_i\left(\sqrt{1+c^2_2\|\v_r[n]\|^4} -c_2\v_r[n]\|^2\right)^{\frac{1}{2}} + P_{f,j}[n]}} \nonumber\\
\mathrm{s.t.}&\quad\eqref{uirs_flight_conds},~\sum^K_{k=1}\alpha_k[n]\underset{m\in\W}{\min}~\left\{1-C_{m,n}\frac{\exp(-\kappa(\|\q_r[n] - q_a\| + \|\q_r[n] - \q_k\|))}{\|\q_r[n] - q_a\|^\rho\|\q_r[n] - q_k\|^\rho}\right\} \geq 1-\varepsilon,~\forall n\in\N, m\in\W \label{p5_cond2}
\end{align}
\end{subequations}
\noindent\rule{\textwidth}{.5pt}
\end{figure*}
where 
\[A_{k,n} = \frac{\alpha_k[n]W}{\ln 2},~B_{k,n} = \frac{p_a[n]L^2}{0.5\hat{p}_j[n]G_{jk}[n]+\sigma^2_k}\left(\frac{\lambda_c}{8\pi\sqrt{\pi} }\right)^2\]
\[C_{m,k,n}= \frac{p_a[n]L^2}{\hat{p}_j[n]G_{jm}[n]}\left(\frac{\lambda_c}{8\pi\sqrt{\pi}}\right)^2,~\forall n\in\N, k\in \K, m\in \W\]
Problem $(\P5)$ is nonconvex  being in the form of min-of-sum-of-fractional-nonlinear-programming. 
Similar to the approach taken to deal $(\P4)$, we can tackle $(\P5)$ by converting the FP to the quadratic subtraction form. To this end, we first mention a fruitful lemma, given below.
\begin{lemma} \label{cvx_log_frac_exp}
Let $\Upsilon(x,y;a,b,c) = \ln\left(1+a\frac{\exp(-b(x + y))}{x^cy^c}\right)$ in the domain of $x,y> 0$ with constants $a,b,c \geq 0$ be a bi-variate function. At a given local point $(x^{lo}, y^{lo})$, the following inequality holds with tightness.
\begin{align}\label{upsilon_lb}
    &\Upsilon(x,y;a,b,c)\geq \ln\left(1+a\frac{\exp(-b(x^{lo} + y^{lo}))}{(x^{lo})^c(y^{lo})^c}\right) \nonumber\\
    &-\frac{a\,\left(c+b\,x^{lo}\right)(x-x^{lo})}{x^{lo}\,\left(a+(x^{lo})^c\,(y^{lo})^c\,{\mathrm{e}}^{b\,\left(x^{lo}+y^{lo}\right)}\right)} \nonumber\\
    &-\frac{a\,\left(c+b\,y^{lo}\right)(y-y^{lo} )}{y^{lo}\,\left(a+(x^{lo})^c\,(y^{lo})^c\,{\mathrm{e}}^{b\,\left(x^{lo}+y^{lo}\right)}\right)}\nonumber\\
    &\treq \Upsilon^{lb}(x,y;x^{lo},y^{lo},a,b,c),
\end{align}
\begin{proof}
We commence with computing the hessian of $\Upsilon(x,y)$, i.e., $\nabla^2 \Upsilon(x,y)$. Since $\Upsilon(x,y)$ is symmetric, we shall only calculate the second order derivatives as

\begin{align}
    \pdv[2]{\Upsilon(x,y)}{x} &= \frac{
    a^2c + a(xy)^c\exp(b (x + y))\left((bx+c)^2 + c\right)}{x^2\left(a + (xy)^c\exp(b (x + y))\right)^2},\\
     \pdv[2]{\Upsilon(x,y)}{x}{y} &= \pdv[2]{\Upsilon(x,y)}{y}{x} = \frac{a(xy)^c(c+bx)(c+by)}{xy(a + (xy)^c\exp(b(x+y)))^2},
\end{align}
Next, we derive determinant of $\nabla^2 \Upsilon(x,y)$ as
\begin{align}
    &\det(\nabla^2 \Upsilon(x,y)) \hspace{-1mm}=\hspace{-1mm} \pdv[2]{\Upsilon(x,y)}{x} \pdv[2]{\Upsilon(x,y)}{y} \hspace{-1mm}-\hspace{-1mm}  \pdv{\Upsilon(x,y)}{x}{y} \pdv{\Upsilon(x,y)}{y}{x},\\
    &\hspace{-3mm}=\hspace{-1mm}\frac{a^2c(ac \hspace{-1mm}+\hspace{-1mm} (xy)^c \exp(b(x\hspace{-1mm}+\hspace{-1mm}y)))\left[2c^2 \hspace{-1mm}+\hspace{-1mm} b^2(x^2\hspace{-1mm}+\hspace{-1mm}y^2) \hspace{-1mm}+\hspace{-1mm} c(2b(x\hspace{-1mm}+\hspace{-1mm}y)\hspace{-1mm}+\hspace{-1mm}1)\right]}{(xy)^2(a \hspace{-1mm}+\hspace{-1mm} (xy)^c \exp(b(x\hspace{-1mm}+\hspace{-1mm}y)))},
\end{align}
It can be readily verified that $\pdv[2]{\Upsilon(x,y)}{x} \geq 0$ and  $\mathbf{det}(\nabla^2 \Upsilon(x,y)) \geq 0$, implying that the hessian matrix of $\Upsilon(x,y)$ is positive semi-definite, i.e., 
\begin{align}
    \nabla^2 \Upsilon(x,y) = \begin{bmatrix}
 \pdv[2]{\Upsilon(x,y)}{x} & \pdv[2]{\Upsilon(x,y)}{x}{y}\\
\pdv[2]{\Upsilon(x,y)}{y}{x}  & \pdv[2]{\Upsilon(x,y)}{y}
\end{bmatrix} \succeq 0,
\end{align}
Therefore, following the second-order condition law \cite{Boyd2004}, the function $\Upsilon(x,y;a,b,c)$ is convex whose first-order Taylor approximation provides a global lower-bound at the given local point, as in \eqref{upsilon_lb}; thus, the proof is completed.
\end{proof}
\end{lemma}
By taking slack block variables $\x=\{x[n], \forall n\}$ and $\y=\{y_k[n] \forall k, n\}$, $\s=\{s_k[n], \forall k, n\}$, $\u=\{u[n],~\forall n\}$, and leveraging Lemmas \ref{lemma2:sum-of-ratios} and \ref{cvx_log_frac_exp}, $(\P5)$ can be roughly restated as
\begin{subequations}
\begin{align}
(\P5.1):& \stackrel{}{\underset{\q_r,\v_r, \x, \y, \u, \s}{\mathrm{maximize}}~~~\eta_r}\nonumber\\
\mathrm{s.t.}&\quad\begin{aligned}
\frac{1}{N}&\sum^N_{n=1}\left(2\gamma_{k,n}f^{lb}_5 - \gamma^2_{k,n}  g^{up}_5(\v_r[n],t[n])\right)\\
&\geq\eta_r,~\forall k\in\K \label{P51_cond1}\\
\end{aligned}\\
\hspace{-10mm}&\eqref{uirs_flight_conds},~\sum^K_{k=1}\alpha_k[n]s_k[n] \geq 1-\varepsilon,~\forall k\in\K, n\in\N \label{P51_cond2}\\
\hspace{-10mm}&\begin{aligned}\label{P51_cond3}
&\ln\left(\frac{1-s_k[n]}{C_{m,k,n}}\right) + \rho \left[\ln(x[n]) + \ln(y_k[n])\right] \\
&+\kappa(x[n] + y_k[n]) \geq 0,~\forall m\in\W, k\in\K, n\in\N\\ 
\end{aligned}\\
\hspace{-10mm}&x[n] \geq \|\q_r[n] - q_a\|,~\forall n\in\N \label{P51_cond4} \\
\hspace{-10mm}&y_k[n] \geq \|\q_r[n] - q_k\|,~\forall k\in\K, n\in\N \label{P51_cond5}\\
\hspace{-10mm}&u^2[n] + 2c_2\|v_r[n]\|^2 \geq \frac{1}{u^2[n]},~\forall n\in\N \label{P51_cond6}\\
\hspace{-10mm}&\|\q_r[n] - \q^{lo}_j[n]\| \geq D_s,~\forall n\in\N \label{P51_cond7}
\end{align}
\end{subequations}
where
\begin{align*}
    f^{lb}_5\hspace{-1mm}=\hspace{-1mm} \sqrt{A_{k,n}\Upsilon^{lb}(x[n],y_k[n];x^{lo}[n],y^{lo}_k[n],B_{k,n},\kappa,\rho)}
\end{align*}
\begin{align*}
    g^{up}_5(\v_r[n],u[n]) &= P_o\left(1+c_0\|\v_r[n]\|^2\right) \nonumber\\
    &+ c_1 \|\v_r[n]\|^3  +  P_i u[n] + P_{f,j}[n],
\end{align*}
The objective function of problem $(\P5.1)$ is now concave, since the square root of an affine function is concave and non-decreasing, therefore each term of summation is concave being in the form of concave-minus-convex, and owing to the fact that sum preserves convexity and the min function is concave and non-decreasing on every argument, thus the combination function is concave. Nonetheless, some nonconvex constraints are introduced with the above reformulation, i.e., \eqref{P51_cond6} and \eqref{P51_cond7}, leading problem $(\P5.1)$ to be non-convex.  Hence, by substituting these nonconvex constraints by their convex approximates, we can achieve the convex reformulation of $(\P5.1)$ represented by
\begin{subequations}\label{subprob_trjUIRS_cvx}
\begin{align}
(\P5.2):& \stackrel{}{\underset{\eta_r, \q_r,\v_r, \x, \y, \u, \s}{\mathrm{maximize}}~\eta_r}\\
\hspace{-10mm}\mathrm{s.t.}&\quad\eqref{P51_cond1}, \eqref{P51_cond2}, \eqref{P51_cond3}, \eqref{P51_cond4}, \eqref{P51_cond5} \label{P52_cond1}\\
&\hspace{-10mm} \begin{aligned}\label{P52_cond2}
&-(u^{lo}[n])^2 +2u[n]u^{lo}[n] + 2c_2\Big(-\|\v^{lo}_r[n]\|^2 \\
&+ 2(\v^{lo}_r[n])^\mathrm{T}\v_r[n]\Big) \geq  \frac{1}{u^2[n]},~\forall n\in\N \\
\end{aligned}\\
\hspace{-10mm}& \begin{aligned}\label{P52_cond3}
&2(\q^{lo}_r[n]-\q^{lo}_j[n])^\mathrm{T}(\q_r[n]-\q^{lo}_j[n]) \\
& - \|\q_r[n]-\q^{lo}_j[n]\|^2  \geq D^2_s,~\forall n\in\N\\
\end{aligned} 
\end{align}
\end{subequations}
where $(\q^{lo}_j, \q^{lo}_r, \v^{lo}_r, \u^{lo})$ is the given local point. Since $(\P5.2)$ is convex, it can be efficiently solved via interior-point method using CVX, providing an approximate solution to problem $(\P5)$, which is summarized in Algorithm \ref{Trj_UIRS_algo}.

\begin{figure}[!t]
 \removelatexerror
 \resizebox{\columnwidth}{!}{%
 \begin{algorithm}[H]
\SetAlgoLined 
\small
\KwResult{$\q^\star_r$, $\v^\star_r$}
\textbf{Initialize} feasible point $(\q^{(0)}_r, \v^{(0)}_r, \x^{(0)}, \y^{(0)}, \u^{(0)})$, set iteration index $t=0$,\\
define $\pmb{\gamma} \treq \{\gamma_{k,n},~\forall k, n\}$ such that
$\gamma_{k,n}\treq\frac{\sqrt{A_{k,n}\Upsilon(x^{(0)}[n],y^{(0)}_k[n];B_{k,n},\kappa,\rho)}}{g^{cvx}_5(\v^{(0)}_r[n],u^{(0)}[n])},$
and let $\eta^{(0)}_r = \underset{k\in\K}{\min}~\frac{1}{N}\sum^N_{n=1}
\frac{A_{k,n}\Upsilon(x^{(0)}[n],y^{(0)}_k[n];B_{k,n},\kappa,\rho)}{g^{cvx}_5(\v^{(0)}_r[n],u^{(0)}[n])}$\\
 \While{\textit{true}}{
   t $\gets$ t + 1\\
   Given $\left(\pmb{\gamma}, \q^{(t-1)}_r, \v^{(t-1)}_r, \x^{(t-1)}. \y^{(t-1)}, \u^{(t-1)}\right)$, solve $(\P5)$ using  \eqref{subprob_trjUIRS_cvx}\\
   Obtain $\left(\q^{(t)}_r, \v^{(t)}_r, \x^{(t)}. \y^{(t)}, \u^{(t)}, \eta^{(t)}_r\right)$ \\
   Update quadratic conversion coefficients as $\pmb{\gamma} \gets \frac{\sqrt{A_{k,n}\Upsilon^{lb}(x^{(m)}[n],y^{(m)}_k[n];x^{(m)}[n],y^{(m)}_k[n],B_{k,n},\kappa,\rho)}}{g^{cvx}_5(\v^{(m)}_r[n],u^{(m)}[n])}$ \\
   \If{$\Big\|\frac{\eta^{(t)}_r - \eta^{(t-1)}_r}{\eta^{(t-1)}_r}\Big\| \leq \epsilon_r$}{
   $(\q^\star_r , \v^\star_r) = (\q^{(t)}_r, \v^{(t)}_r)$;\\
   \textit{break};}
  
}
 \caption{\label{Trj_UIRS_algo} \small Proposed algorithm to approximately solve subproblem $(\P5)$} 
\end{algorithm}}
\end{figure}

\subsection{Overall Algorithm}
Now we are ready to put together all the subsolutions and propose a low-complex sequential algorithm (see Algorithm \ref{myalgorithm_overal}) to improve mAEE of the considered scenario subject to  the covertness requirement.
\begin{figure}[!t]
 \removelatexerror
 \resizebox{\columnwidth}{!}{%
  \begin{algorithm}[H]
  \small
  \caption{\small Overall proposed iterative algorithm for mAEE optimization}\label{myalgorithm_overal}
  1:~\textbf{Initialize}~
  a feasible point $\left(\pmb{\alpha}^{(0)},\Pow^{(0)}_a, \pmb{\Phi}^{(0)}, \Pow^{(0)}_j, \Q^{(0)}_r, \Q^{(0)}_j, \pmb{\Phi}^{(0)}\right)$, and set iteration index $i=0$;\\
  2:~\textbf{Repeat:}~ $i \gets i+1$;\\
  3:~Solve $(\P1)$ using \eqref{usrsch_subprob_cvx}, updating $\pmb{\alpha}^{(i+1)}$;\\
  4:~Given $\pmb{\alpha}^{(i+1)}$, solve $(\P2)$ using \eqref{jointpow_cvx}, updating $\Pow^{(i+1)} = \{\Pow^{(i+1)}_a, \hat{\Pow}^{(i+1)}_j\}$\\
  5:~Given $\left(\pmb{\alpha}^{(i+1)}, \Pow^{(i+1)}\right)$, solve $(\P4)$ via Algorithm \ref{Trj_UCJ_algo}, updating $\Q^{(i+1)}_j = \{\q^{(i+1)}_j, \v^{(i+1)}_j\}$\\
  6:~Given $\left(\pmb{\alpha}^{(i+1)}, \Pow^{(i+1)}, \Q^{(i+1)}_j\right)$, solve $(\P5)$ via Algorithm \ref{Trj_UIRS_algo}, updating $\Q^{(i+1)}_r = \{\q^{(i+1)}_r, \v^{(i+1)}_r\}$\\
7:~Given $\left(\Q^{(i+1)}_r, \Q^{(i+1)}_j\right)$, calculate  $\pmb{\Phi}^{(i+1)}$ using \eqref{IRSbeamforming_opt}\\
  8: Given $\left(\pmb{\alpha}^{(i+1)}, \Pow^{(i+1)}, \Q^{(i+1)}_j, \Q^{(i+1)}_r, \pmb{\Phi}^{(i+1)}\right)$, calculate the value of mAEE in \eqref{opt_prob}\\ 
  9:~\textbf{Until} fractional increase of mAEE gets below the terminating threshold $\epsilon_o$
  \end{algorithm}}
\end{figure}

\subsection{Complexity and Convergence Analysis}
It can be mathematically proved that both the inner FP optimization and outer alternating optimization of the proposed BSCA-based algorithm are guaranteed to converge with any feasible initialization. Since the feasible solution set of $\mathrm{(P)}$ is compact, its objective value is non-decreasing over iteration index, and that the optimal value of mAEE is upper bounded by a finite value. The detail is omitted for brevity, but the interested readers can refer to \cite{TatarMamaghani2021, FP2018Shen}.  Further, letting $\gamma_1$, $\gamma_4$, $\gamma_5$ be the maximum iteration number required for convergence of subproblems (\P1.1), (\P4.2), and, (\P5.2), and then the worst-case computational complexity for each sub-problem and the complexity of proposed overall algorithm are obtained in Table \ref{table:complexity}. It is worth noticing that the overall Algorithm \ref{myalgorithm_overal}'s complexity is in polynomial time order, demonstrating the efficiency of our proposed algorithm and suitability to the energy-constrained UAV-IoT scenarios.
\begin{table}[t]
\caption{Complexity analysis}
\centering
\resizebox{\columnwidth}{!}{%
\begin{tabular}{l c} 
 \hline  \hline
 Problem & Complexity \\ [0.5ex] 
 \hline 
 User Scheduling Subproblem (\P1.1) & $\O\left(\gamma_1(NK+1)(2NK+3N+K+2)^\frac{3}{2}\right)$\\
 Joint Power Allocation Subproblem (\P2.2) & $\O\left((N(K+2)+1)^2(NK(K-1)+3N+K+1)^\frac{3}{2} \right)$\\
 SDP relaxation IRS Subproblem (\P3.2.1) & $\O\left(N^4(K+L+2)^4(N(K+L))^{\frac{1}{2}} \right)$\\
  RM-SCA for IRS optimization (\P3.3) & $\O\left(\log({\mu_{max}/\varrho})(N(K+L+4)+K)^4(N(K+L+1)+1)^{\frac{1}{2}} \right)$\\
  Alternative IRS beamforming Design \eqref{IRSbeamforming_opt} & $\O\left(LN\right)$\\
   UCJ's trajectory design Subproblem (\P4.2)&$\O\left(\gamma_4(2NK+5N+1)^2(NK^2+2NK+6N+K)^\frac{3}{2}\right)$\\
   UIRS's trajectory design Subproblem (\P5.2)&$\O\left(\gamma_5(2NK+6N+1)^2(NK^2+3NK+7N+K)^\frac{3}{2}\right)$\\
     Overall Complexity of Algorithm \ref{myalgorithm_overal} & $\O\left((\gamma_4+ \gamma_5) N^\frac{7}{2}K^5 + \gamma_1(NK)^\frac{5}{2} + LN \right)$
\\[1ex] 
 \hline  \hline
\end{tabular}}
\label{table:complexity}
\end{table}

\section{Numerical results and discussion}\label{sec:numerical}
In this section, we evaluate the performance of our proposed optimization algorithm to improve the mAEE performance metric for the considered UIRS-assisted THz covert communication  system with UCJ's cooperative jamming. To demonstrate the effectiveness of our design, we compare it with some benchmarks listed below.

\begin{itemize}
      \item  \textbf{Proposed JTCD}: Proposed joint trajectory and communication design for the mAEE improvement according to Algorithm \ref{myalgorithm_overal}.
    
     \item \textbf{Benchmark I - CD}: With fixed UAVs' trajectory and velocity block variables, only communication design is taken into account including user scheduling and transmit powers optimization $(\mathbf{P}^\star_a, \mathbf{P}^\star_j, \pmb{\alpha}^\star)$.
     
     \item \textbf{Benchmark II - TD}: Keeping the communication resources fixed,  joint trajectory and velocity design of both UAVs are considered  $(\mathbf{Q}^\star_r=\{\q^\star_r, \v^\star_r\},  \mathbf{Q}^\star_j=\{\q^\star_j, \v^\star_j\})$.
      
     \item \textbf{Benchmark III - IFTR}: Non-optimal initial feasible trajectory and resource allocations.
\end{itemize}

Unless otherwise stated, all simulation parameters are set as given follows: $v^{max}_r=v^{max}_j=25$ m.s\textsuperscript{-1}, $a^{max}_r=v^{max}_j=6$ m.s\textsuperscript{-2}, \textcolor{black}{$\rho = 2.3$}, $f_c = 0.3$ THz, $\kappa = 3.2094\times10^{-4}$ using \eqref{kappa_eq}, $p^{max}_a = p^{max}_j = 1$ W, $p_{tot} = 40$ W, $\sigma^2_k = -180,~\forall k$, $\q^a =[$0$, $0$, $0$]$, $\q^I_r =[$100$, $0$, $50$]$, $\q^I_j =[$80$ , $0$, $50$]$, $P_o = 79.856$ W and $P_i = 88.63$  \cite{Zeng2019b}, $K=5$ UEs are uniformly distributed inside a disk centered at $\q_a$ with inner and outer radii $R_i=100$ m and \textcolor{black}{$R_o=200$ m}, respectively, $T = 30$ s, $\delta_t = 0.1$ s, $(\epsilon_o, \epsilon_r, \epsilon_j) = (10^{-1}, 10^{-3}, 10^{-3})$. $L_x=5$, $L_y = 6$, $\delta_x = \delta_y = 1$ mm, 
$(c_0, c_1, c_2) = (2.0833\times 10^{-4}, 0.0092,  0.0308)$ \cite{long2020reflections}.

The initial feasible trajectory and velocity of UAVs are obtained as in \cite{mamaghani2021terahertz}. UEs are randomly selected such that user selection constraint $\C1$ is satisfied, and $p_a[n] = \frac{p_{tot}}{4T}, \hat{p}_j[n] = \frac{p_{tot}}{2T}, \forall n$ are set fixed  values such that constraints $\C3$ and $\C5$ are met initially. Note that one can readily obtain a feasible point of $\mathbf{P}$ by setting AP's transmit power to zero (or some non-zero but very small value to avoid numerical problems) and maximum AN's transmit powers to the average of the network power budget. 

\begin{figure*}[t!]
 \centering
     \begin{subfigure}[t]{\columnwidth}
        \centering
        \includegraphics[width=\textwidth]{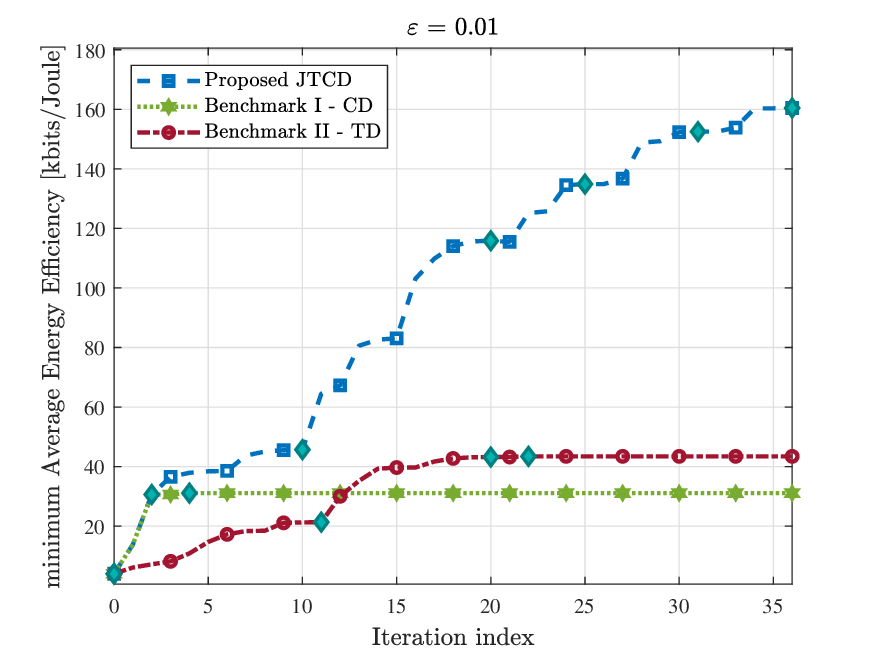}
    \end{subfigure}
    ~ 
 \begin{subfigure}[t]{\columnwidth}
        \centering
        \includegraphics[width= \textwidth]{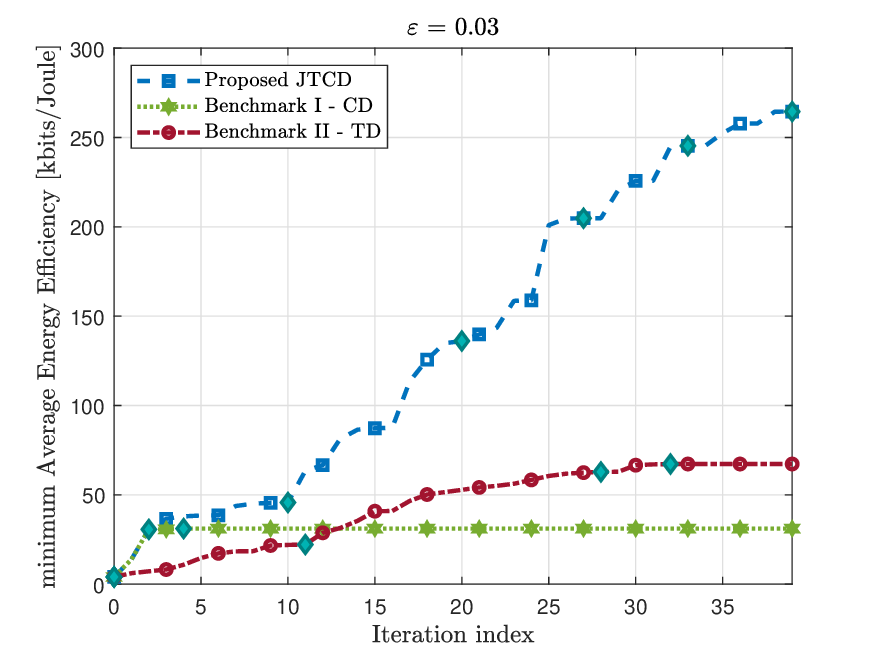}
\end{subfigure}
     \caption{ \label{sim:figs12} mAEE vs Iteration index.} 
\end{figure*}
In Fig. \ref{sim:figs12}, we plot mAEE versus iteration indices for different levels of covertness to demonstrate the convergence of the proposed iterative algorithm and verify the correctness of the analysis, i.e., the objective value of the optimization problem is non-decreasing over the iteration index. Note that the objective value obtained following the outer iterations are marked as blue diamonds on the curves. 
Although both the CD and TD schemes have fast convergence, the superiority of the JTCD  compared to the other benchmarks is crystal clear; specifically, the proposed JTCD reaches the mAEE performance approximately quadruple as high as the CD counterpart at the end of the overall $6$ iterations.  We also observe from Fig. \ref{sim:figs12} that as $\varepsilon$ increases, i.e.,  the covertness constraint gets loosened, the adopted trajectory and resource allocations can be adjusted appropriately to acquire significantly higher mAEE for all schemes.



  \begin{figure}
  \centering
  \begin{subfigure}{0.45\columnwidth}
  \includegraphics[width=\textwidth]{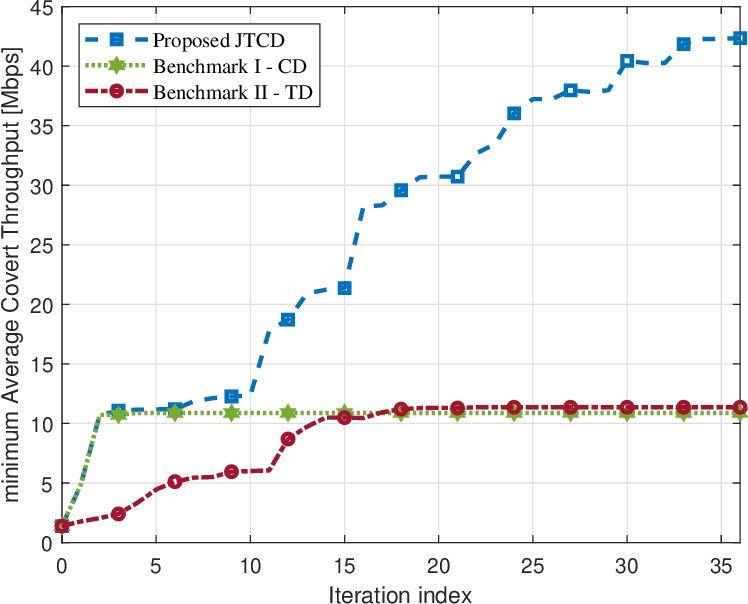}
  \end{subfigure}
  \hfill
  \begin{subfigure}{0.45\columnwidth}
  \includegraphics[width=\textwidth]{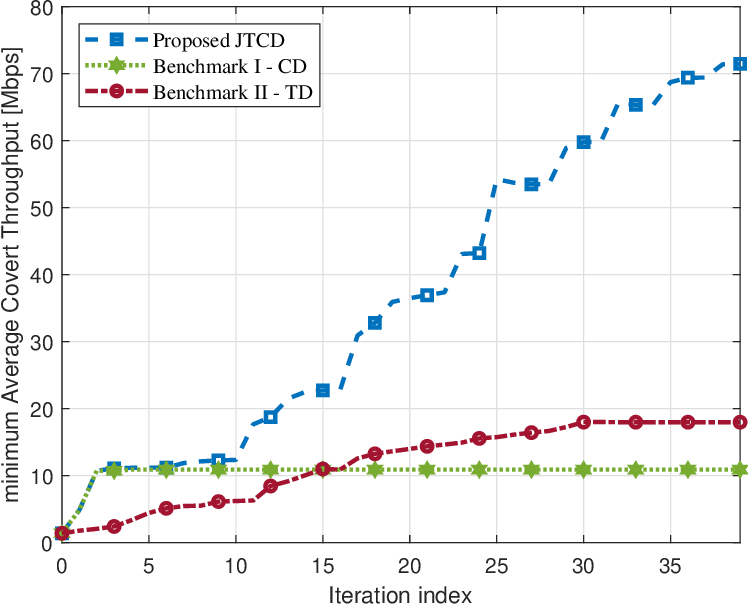}
  \end{subfigure} 
  \begin{subfigure}{0.45\columnwidth} 
  \includegraphics[width=\textwidth]{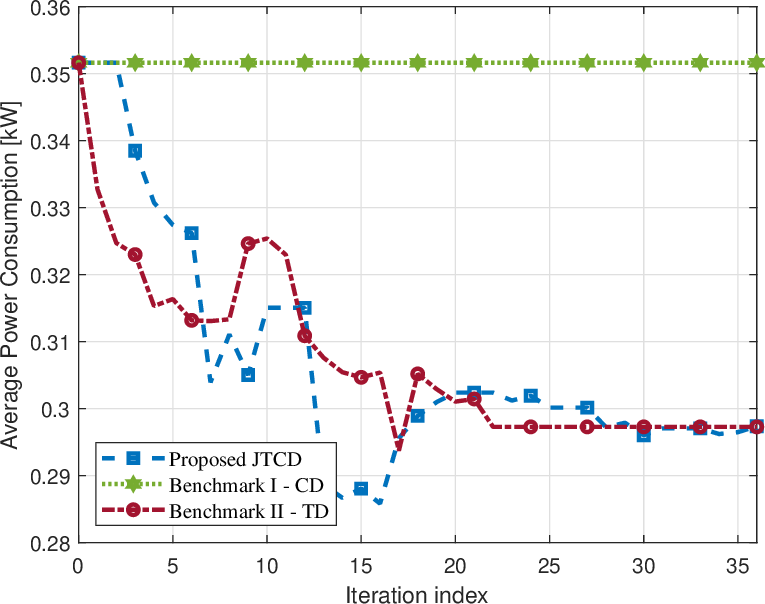} 
  \caption{$\varepsilon = 0.01$.} 
  \end{subfigure}  
  \hfill 
  \begin{subfigure}{0.45\columnwidth} 
  \includegraphics[width=\textwidth]{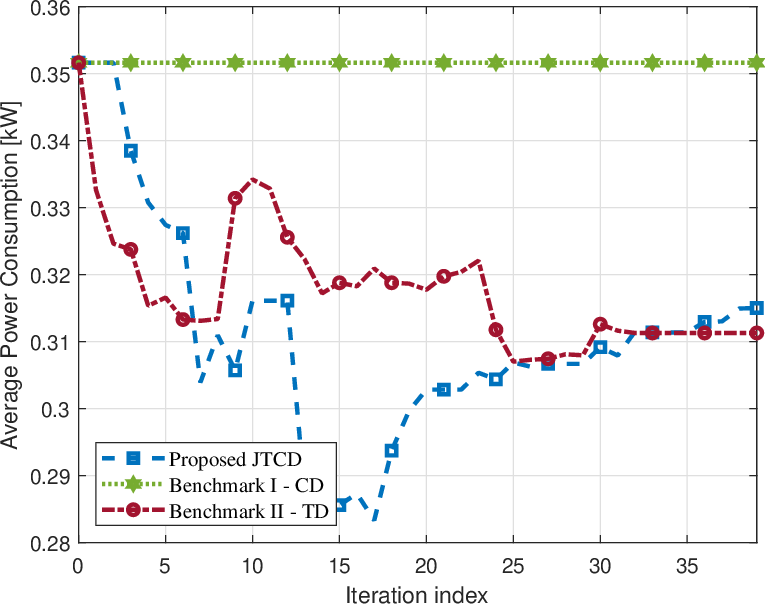} 
  \caption{ $\varepsilon = 0.03$.} 
  \end{subfigure}
     \caption{\label{sim:figs34} \textcolor{black}{mACT and APC vs Iteration index.}} 
  \end{figure}

\textcolor{black}{Fig. \ref{sim:figs34} illustrates how the minimum average covert throughput (mACT) and the average power consumption (APC) arising from the UAVs propulsion vary over iteration numbers with different covertness constraints, i.e., $\varepsilon = \{0.01, 0.03\}$. Depending on the level of required covertness, we observe that the TD scheme can achieve comparably higher mACT performance than that of CD. Specifically, TD can achieve $11.36$ Mbps with the covertness requirement of $\varepsilon = 0.01$ and $17.96$ Mbps with $\varepsilon = 0.03$, while CD obtains relatively similar mACT performance, i.e., $10.88$ Mbps, regardless of the level of covertness, according to our setup. This implies the significance of the trajectory optimization and performance improvement brought by the flexible 3D network design. Further, however, the covertness level increases (or equivalently, the value of $\varepsilon$ decreases), the mACT decreases due to stricter constraints. Hence, we can observe from the figure that the benefits of the collaborative design as shown in JTCD curves are well pronounced compared to the benchmark schemes. Last but not least, the APC follows, as expected, an overall decreasing, but not monotonically, trend over the iterations for both $\varepsilon = 0.01$ and $\varepsilon = 0.03$ in order to improve the mAEE.}




\begin{figure}[t]
\centering
\includegraphics[width=1.2 \columnwidth]{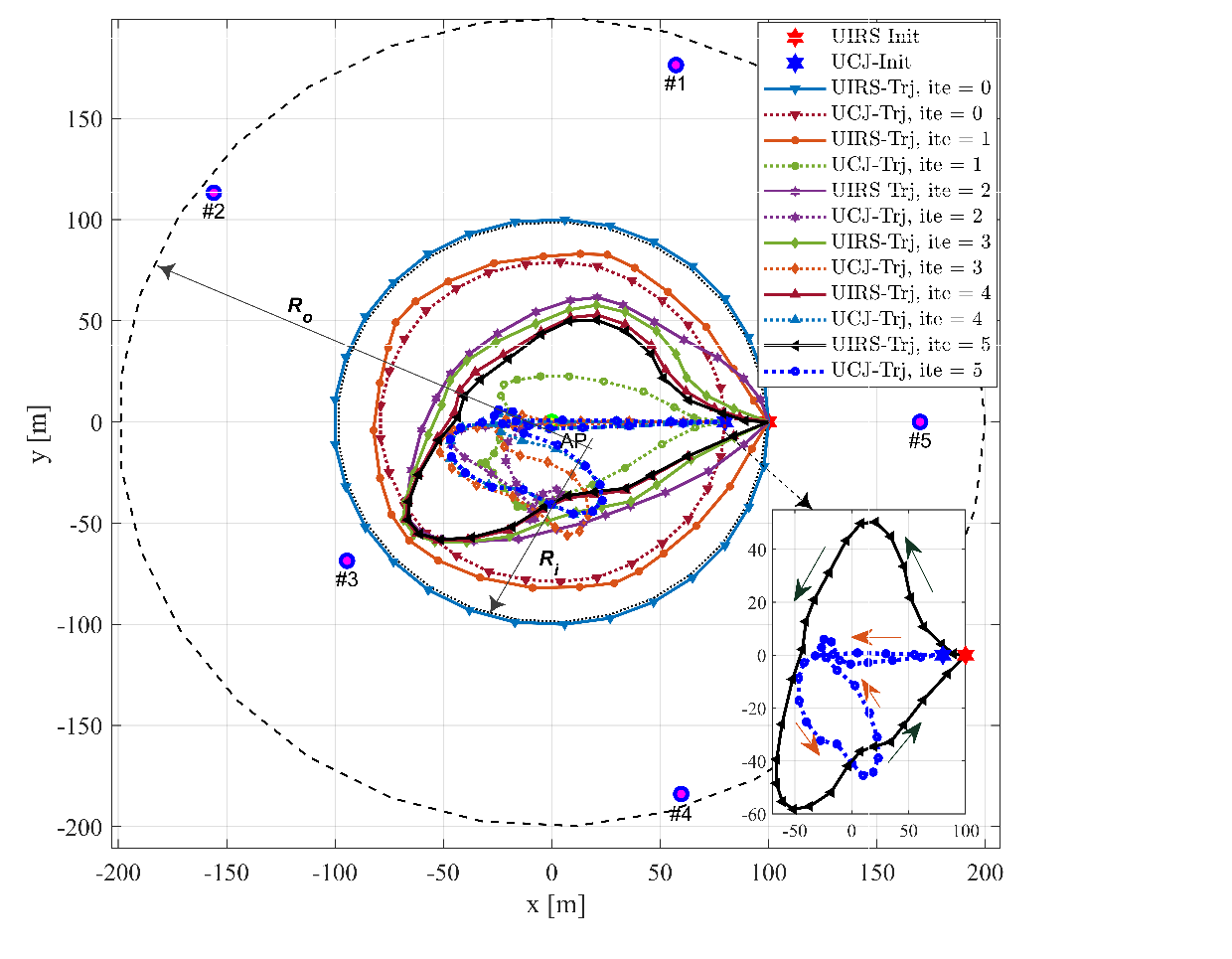}
\caption{\small \textcolor{black}{UIRS and UCJ's optimized trajectories according to the Proposed JTCD scheme.}}
\label{sim:fig7}
\end{figure}

\begin{figure}[t]
\centering
\includegraphics[width=1.2\columnwidth]{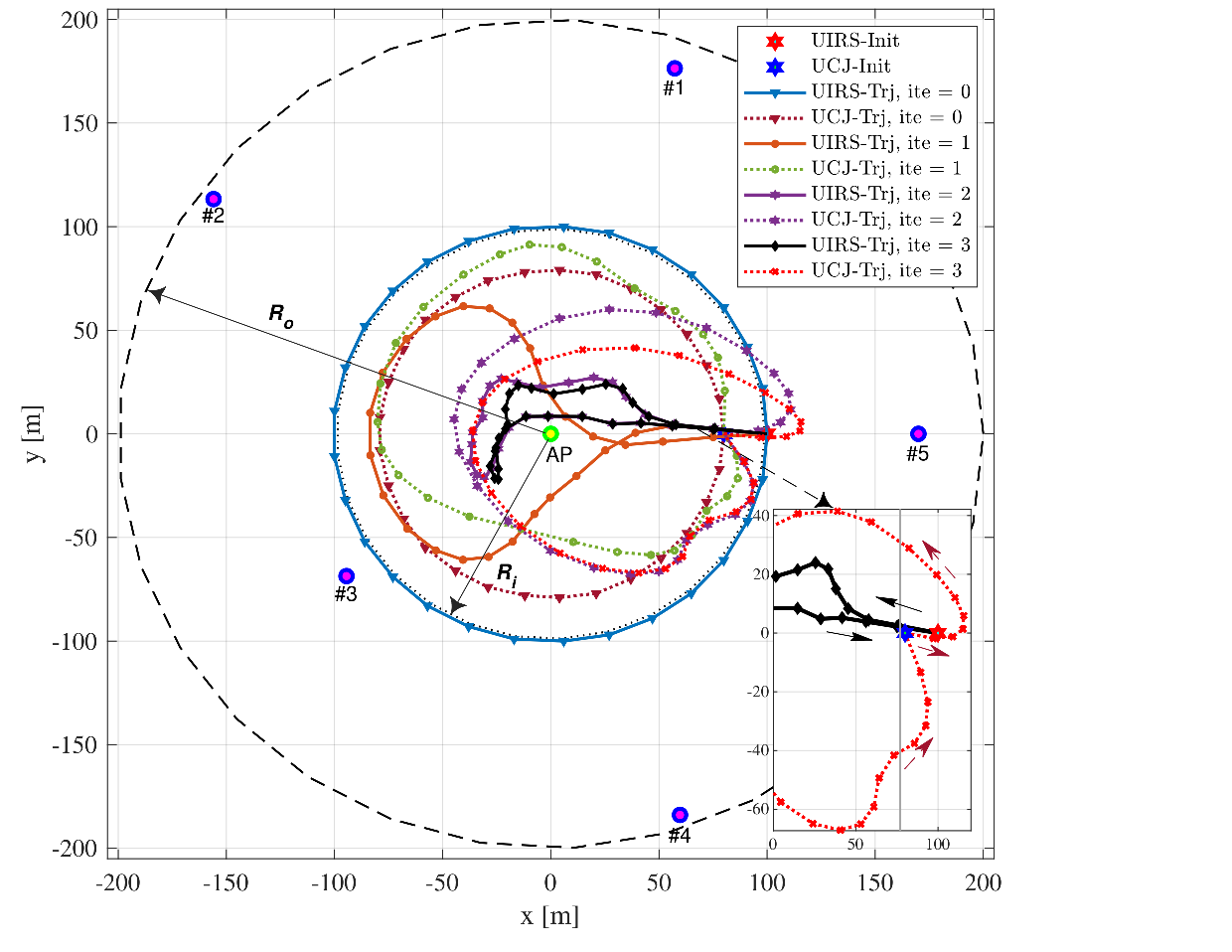}
\caption{\textcolor{black}{\small UIRS and UCJ's optimized trajectories according to Benchmark II.}}
\label{sim:fig8}
\end{figure}

\begin{figure}[t!]
 \centering
    \begin{subfigure}[t]{\columnwidth}
        \centering
        \includegraphics[width= \columnwidth]{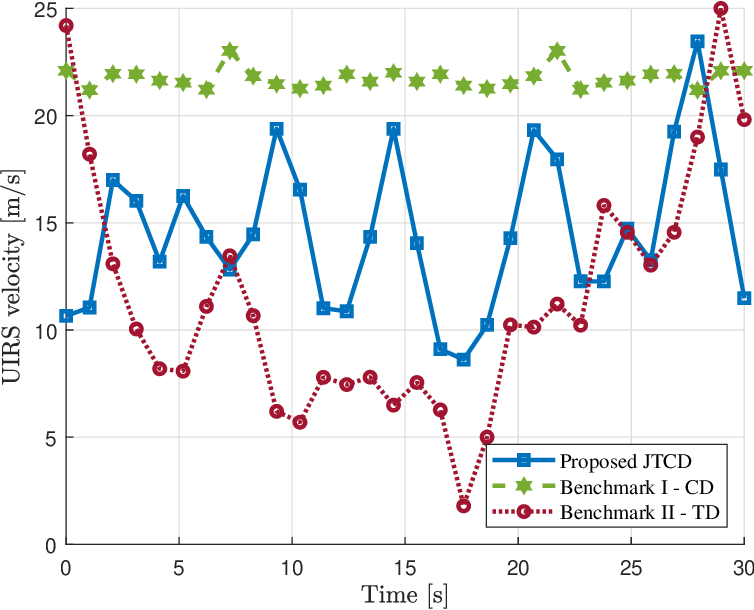}
        \caption{UIRS  optimized velocity.}
        \label{sim:fig9}
    \end{subfigure} 
    ~ 
 \begin{subfigure}[t]{\columnwidth}
        \centering
        \includegraphics[width= \columnwidth]{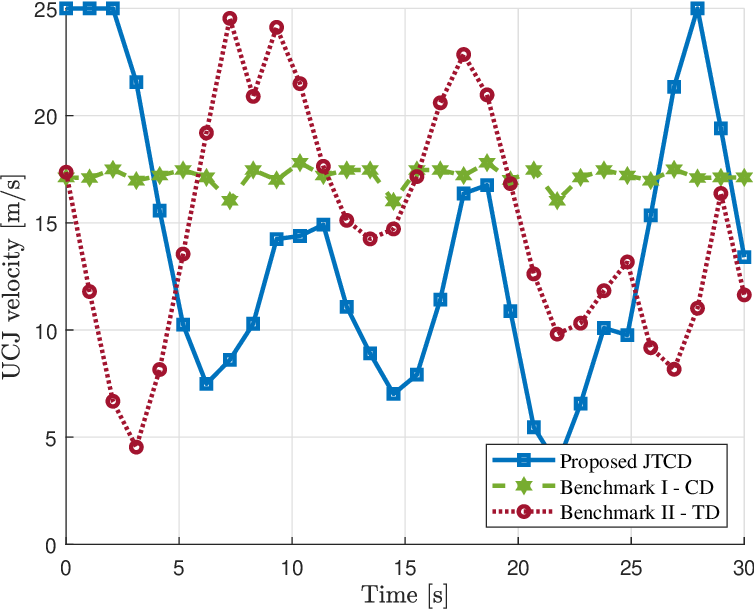}
        \caption{UCJ optimized velocity.}
        \label{sim:fig10}
\end{subfigure}
 \caption{UIRS and UCJ's velocity according to different schemes with $T=30$s and $\varepsilon = 0.01$.} 
\end{figure}

In Figs. \ref{sim:fig7} and  \ref{sim:fig8}, we plot the 2D view of optimized trajectories of UIRS and UCJ using JTCD and TD schemes, respectively. Initial feasible trajectories are labeled as $ite = 0$, which also belong to the CD scheme. The corresponding UIRS and UCJ's velocity are shown in Figs. \ref{sim:fig9} and \ref{sim:fig10}. We observe that over the iteration index, the path planning properly shapes with a relatively complicated form than the initial circular one to improve the system's mAEE performance. The trajectory variation between two consecutive iterations gets negligible as the algorithm approaches convergence.

Notice that UIRS should adjust its path flying towards the best location to obtain AP's data with low transmit power and then beamform towards the scheduled UE Bob at the given time slot while satisfying covertness requirement in terms of transmission detection failure of the strongest unscheduled user Willie. Therefore, for energy-efficient covert data transmission purposes, the UIRS's trajectory initially shrinks relatively towards Bob with the shortest distance, which can be verified in Fig. \ref{sim:fig13}; nonetheless,  adaptively gets adjusted when the UCJ's trajectory is updated.  Further, in Fig. \ref{sim:fig8}, the UCJ's designed path based on TD differs from the one obtained using the JTCD scheme due to a different set of the scheduled UEs over the time slots. However, in both scenarios, UCJ attempts to find a path with the highest detrimental effect on the Willies' detection rate while compensating such effect at Bob for the sake of mAEE improvement. Additionally, we can observe in Figs. \ref{sim:fig9} and \ref{sim:fig10} that maintaining approximately fixed velocity or hovering at some particular locations for a while are not appropriate for energy-efficient design, which is fundamentally different from conventional designs focusing on solely covert data rate enhancement, e.g., \cite{Zhou2021d}. Indeed, both UAVs must follow specific velocity adjustment patterns to reduce the mechanical power consumption while improving the covert throughput.
\begin{figure}[t]
\centering
\includegraphics[width= \columnwidth]{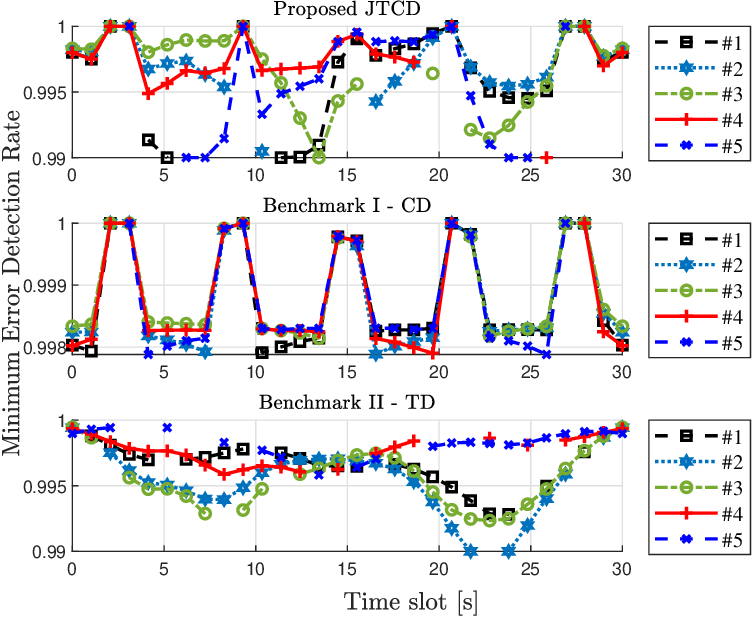}
\caption{\small  $\zeta^{\star}_{m,k}$ and selected UEs according to different algorithms.}
\label{sim:fig13}
\end{figure}
Fig. \ref{sim:fig13} depicts the minimum detection error rate of each UE against the time slot for different algorithms. The absence of some curves at a particular time slot, e.g., UE $\#1$ curve between time slot $T=6$s and $T=11$s inclusive for the JTCD scheme, indicates the scheduled UE Bob. We observe that the minimum detection error rates stratify the covertness requirement for all scenarios, i.e., $\zeta^\star_{m,k} \geq 1 - \varepsilon, \forall m, k$ with $\varepsilon = 0.01$ and the strongest detector Willie can achieve a minimum total detection error rate of no less than the required covertness. Further, with the fixed circular trajectory, we can see that mainly the closest users are scheduled for covert communication according to the CD scheme. However, when it comes to the JTCD scenario, based upon the level of AN transmission and the locations of UIRS and UCJ, a different set of UEs can be selected for energy-efficient covert transmission.

\begin{figure}[t]
\centering
\includegraphics[width= \columnwidth]{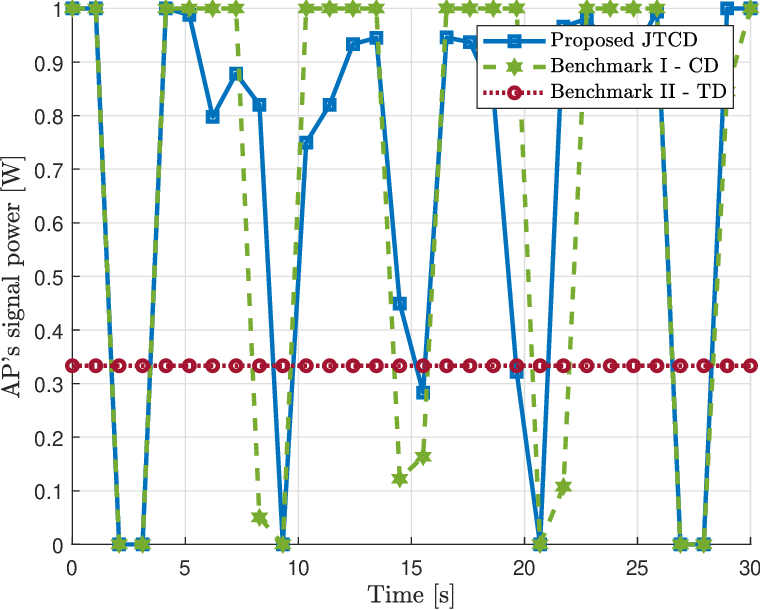}
\caption{\small AP's signal transmission power vs. time.}
\label{sim:fig11}
\end{figure}

\begin{figure}[t]
\centering
\includegraphics[width= \columnwidth]{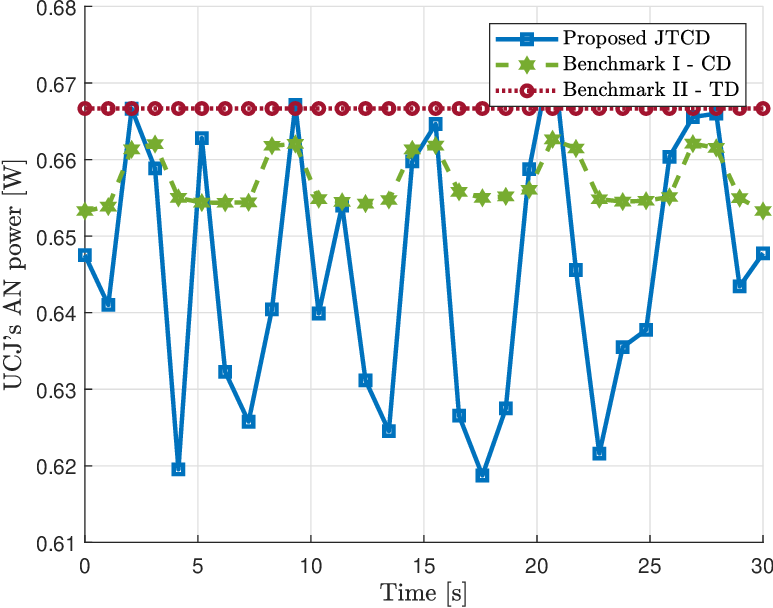}
\caption{\small  Maximum UCJ's AN transmission power vs. time.}
\label{sim:fig12}
\end{figure}
In Figs. \ref{sim:fig11} and \ref{sim:fig12}, we examine the power allocation amongst AP's signal transmission and UCJ's AN transmission vs. time slot. Note that mAEE is an increasing function of AP's signal transmission power, while decreasing as AN transmission power increases. However, by proper resource and trajectory design, the level of maximum AN can be decreased without violating the covertness requirement. As an example, consider power allocation obtained according to Benchmark I with fixed circular UAVs' trajectories. Initially, UE $\#5$ is scheduled as Bob due to closeness to the UIRS according to Fig. \ref{sim:fig13} and so AP transmits signals with its maximum instantaneous power; however, as the UIRS gets farther from Bob or closer to Willie with the strongest detector, i.e., UE $\#1$, AP's power drops while maximum AN transmission gets increased to satisfy the covertness requirement.

In the experiment depicted in Fig. \ref{sim:fig14}, we compare how the mAEE performance varies with the increase of  IRS elements for different scenarios. First, we note that the more the IRS elements, the higher the mAEE performance, owing to the fact that mAEE is an increasing function of $L$. However, by virtue of our worst-case design, an mAEE ceiling occurs at large $L$ due to the covertness constraint. Plus, such an increasing trend can be justified based on the fact that the IRS consists of a passive reflecting structure and does not cost any specific energy whilst improving the covert transmission rate. We see that the proposed JTCD scheme, Benchmark I - CD, and Benchmark II - TD, can increase mAEE by approximately $22$, $9$, and $5$ times,  compared to the initial feasible point (Benchmark III - IFTR).  When solely comparing Benchmark I and II schemes, Benchmark I - CD is the winner. This can be justified that although the trajectory design improves mAEE even with less number of IRS elements, resource allocation with a larger number of IRS elements dominates the final mAEE. Overall, we see that our proposed JTCD scheme outperforms other counterparts significantly for practical numbers of IRS elements. 
\begin{figure}[t]
\centering
\includegraphics[width= 1.1\columnwidth]{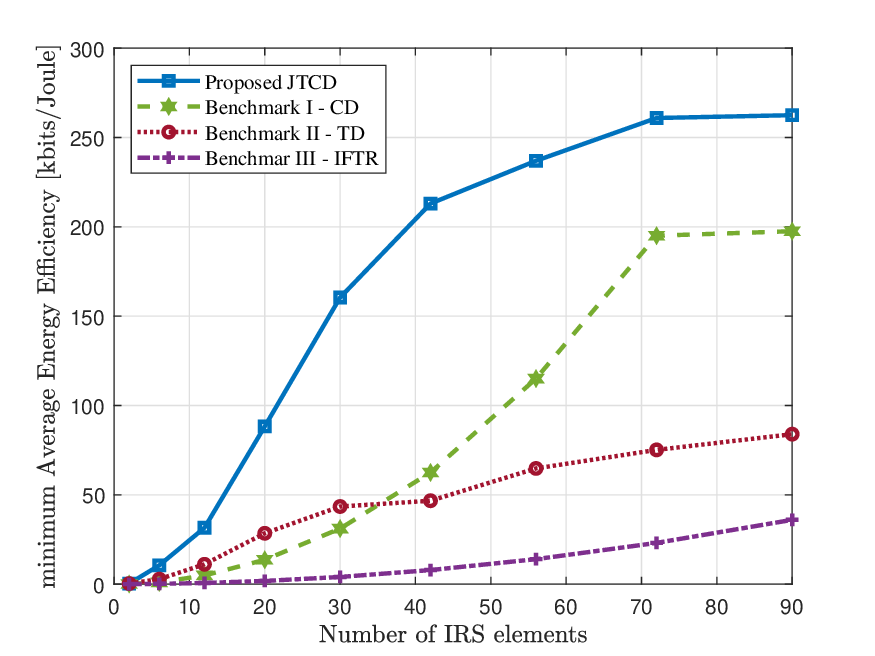}
\caption{\small  Effect of number of IRS elements $L = L_x \times L_y$.}
\label{sim:fig14}
\end{figure}

\begin{figure}[t]
\centering
\includegraphics[width=\columnwidth]{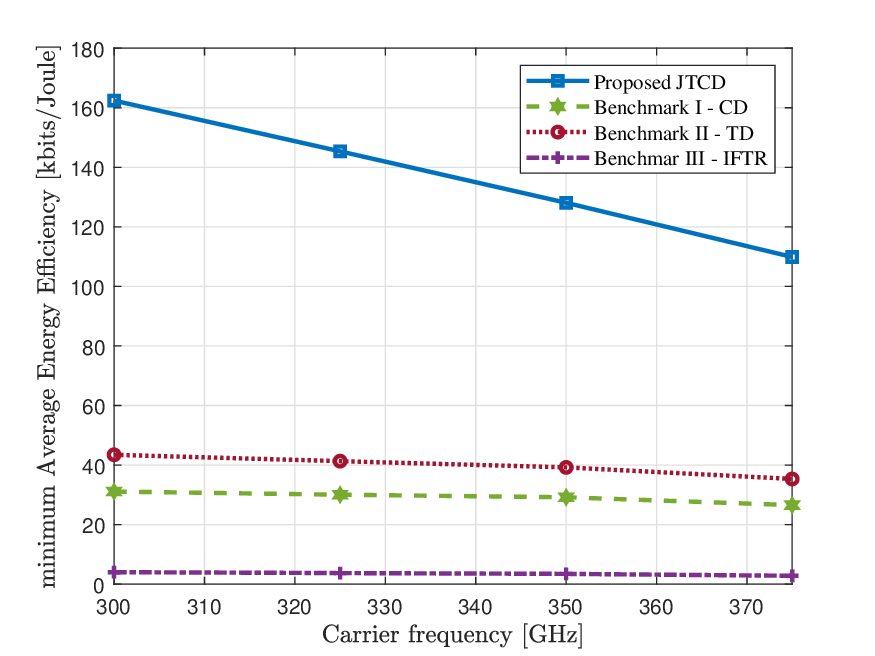}
\caption{\small  Effect of carrier frequency on mAEE.}
\label{sim:fig15}
\end{figure}

Fig. \ref{sim:fig15} illustrates mAEE performance vs. carrier frequency, which also impacts the molecular absorption coefficient of THz links. As carrier frequency increases, more propagation loss and higher molecular absorption occur,  bringing two impacts: one is on the reduced achievable covert rate and degraded mAEE performance, the other is on the diminished strength of the received signal at not only Bob, but also Willies. The former degrades the overall mAEE performance while the latter loosens the covertness constraint.

\section{Conclusions}\label{sec:conclusion}
This article addressed the energy-efficient design of a multi-UAV wireless communication system for covert data dissemination in the B5G UAV-IoT network operating at THz bands. Specifically, the UIRS is used for the passive covert data relaying from the AP to the scheduled UE, and the UCJ is employed for efficiently AN injection, degrading unscheduled UEs' detection performance. We showed analytically that properly employing a THz-operated UIRS-UCJ system can significantly improve communications reliability, capacity, coverage, and covertness, thanks to the dynamic nature of such a system. Then, the transmission detection performance from the perspective of unscheduled UEs as potential adversaries was explored, and the analytical closed-form expressions were derived to evaluate covertness. Further, to improve the overall system performance, we developed a low-complex algorithm to iteratively solve a sequence of convex optimization problems for maximizing the mAEE performance subject to the covertness requirement via the joint design of the user scheduling, network transmission power, IRS beamforming, and UAVs' trajectory planning.  Our examination revealed significant outperformance of our proposed JTCD scheme compared to the other schemes in terms of both mAEE and mACT. Future work may focus on practical THz channel modeling and imperfect IRS phase shift and amplitude design for such a UAV-IoT system.


\appendices
\numberwithin{equation}{section}
\makeatletter 
\newcommand{\section@cntformat}{Appendix \thesection:\ }
\makeatother

\section{Joint design of network power allocation and user scheduling \ref{lemma1:rank1refomulaiton}}\label{Appendix A}

\renewcommand{\w}{\ensuremath{\mathbf{w}}}
\renewcommand{\s}{\ensuremath{\mathbf{v}}}
\renewcommand{\r}{\ensuremath{\mathbf{r}}}
\renewcommand{\C}{\ensuremath{\mathrm{C}}}
Here, we try to jointly optimize network transmission powers and user scheduling variables, i.e., $(\Pow_a, \hat{\Pow}_j, \pmb{\alpha})$. Therefore, by taking slack variables $\w=\{w_k[n], \forall n\in\N, k\in\K\}$, $\s=\{s_k[n], \forall n\in\N, k\in\K\}$ , $\r=\{r[n], \forall n\in\N\}$, we can represent $(\P)$ equivalently as
\begin{subequations}
\begin{align}
(\widetilde{\P3}):& \stackrel{}{\underset{\psi, \pmb{\alpha}, \Pow_a, \hat{\Pow}_j, \w, \s, \r}{\mathrm{maximize}}~\psi } \nonumber\\
\mathrm{s.t.}&\quad\frac{1}{N}\sum^{N}_{n=1}A_n\alpha_k[n]w_k[n] \geq \psi,~\forall k\in\K\label{ph3_cond1}\\
&\hspace{-10mm}\ln\left(1+\frac{B_{k,n} p_a[n]}{C_{k,n} \hat{p}_j[n] + 1}\right) \geq w_k[n],~\forall k\in\K, n\in\N \label{ph3_cond2}\\ 
&\hspace{-10mm} \sum^{K}_{k=1} \alpha_k[n]s_k[n] \geq 1 - \varepsilon,~\forall n\in\N\label{ph3_cond3} \\
&\hspace{-10mm} 1- s_k[n] \hspace{-1mm}\geq\hspace{-1mm} \underset{m \in \W}{\max}\left\{\hspace{-1mm}D_{n,k,m} \frac{p^2_a[n]}{r[n]}\hspace{-1mm}\right\},~\forall k\in\K, n\in\N \label{ph3_cond4}\\
&\hspace{-10mm} r[n] \leq p_a[n] \hat{p}_j[n],~\forall n\in\N \label{ph3_cond5}\\
&\hspace{-10mm}\C4~\&~\C6 \label{ph3_cond6}
\end{align}
\end{subequations}
where $B_{k,n}$, $C_{k,n}$, $D_{n,k,m}$ are defined in $(\P2)$. Further, we have
\[A_n = \frac{\mathrm{W}}{\ln(2)(P_{f,r}[n] + P_{f,j}[n])},~\forall n\in\N\]
We note that quadratic-over-linear is a convex function \cite{Boyd2004} and max of some convex functions is convex too, thus constraint \eqref{ph3_cond4} is convex. Nonetheless, problem $(\widetilde{\P3})$ is not convex due to nonconvex constraints \eqref{ph3_cond1}, \eqref{ph3_cond2}, \eqref{ph3_cond3}, \eqref{ph3_cond5}. To this end, we first mention lemma below to handle such challenging problem.
\begin{lemma}\label{xy}
Define the bi-variate function $f(x,y) = xy$ with $x\neq y$ which is a log-concave function. The global lower-bound and upper-bound of which at the given point $(x^{lo}, y^{lo})$ can be obtained as
\begin{align}\label{xy_lb}
    f(x,y) &\geq \frac{1}{4}\left( -(x^{lo}+ y^{lo})^2 + 2(x^{lo}+ y^{lo})(x+y) - (x-y)^2\right)\nonumber\\
    &\treq f^{lb}_{ccv}(x,y,x^{lo},y^{lo}),
\end{align}
\begin{align}\label{xy_ub}
    f(x,y) &\leq  \frac{1}{4}\left((x+y)^2  +(x^{lo} - y^{lo})^2 - 2(x^{lo}- y^{lo})(x-y)\right)\nonumber\\
    &\treq f^{ub}_{cvx}(x,y,x^{lo},y^{lo}),
\end{align}
\end{lemma}
\begin{proof}
Using the difference of quadratic reformulation, i.e., $xy = \frac{1}{4}\left[(x+y)^2 - (x-y)^2\right]$ and properly employing the first order condition law \cite{Boyd2004} via applying restrictive approximation to the first or second quadratic terms, one can reach the aforementioned bounds.  
\end{proof}

Using Lemma \ref{xy} and the way we tackled binary use scheduling in $(\P1.1)$, we can reformulate $(\widetilde{\P3})$ as a convex optimization problem given by
\begin{subequations}
\begin{align}
(\widetilde{\P3.1}):& \stackrel{}{\underset{\psi, \pmb{\alpha}, \Pow_a, \hat{\Pow}_j, \w, \s, \r, \eta}{\mathrm{maximize}}~\psi - \eta \mu } \nonumber\\
\hspace{-15mm}\mathrm{s.t.}&~\begin{aligned}\label{ph31_cond1}
&\frac{1}{N}\sum^{N}_{n=1}A_nf^{lb}_{ccv}(\alpha_k[n],w_k[n],\alpha^{lo}_k[n],w^{lo}_k[n])\\
&\geq \psi,~\forall k\in\K\\
\end{aligned}\\
&\hspace{-15mm} \begin{aligned}\label{ph31_cond2}
&\ln\left(1+B_n p_a[n] + C_n \hat{p}_j[n]\right)-f_{lb}(\hat{p}_j[n],\hat{p}^{lo}_j[n]) \\
&\geq w_k[n],~\forall k\in\K, n\in\N \\ 
\end{aligned}\\
&\hspace{-15mm} \sum^{K}_{k=1}
f^{lb}_{ccv}(\alpha_k[n],s_k[n],\alpha^{lo}_k[n],s^{lo}_k[n]) \hspace{-1mm} \geq \hspace{-1mm} 1 - \varepsilon,~\forall n\in\N\label{ph31_cond3} \\
&\hspace{-15mm}f^{lb}_{ccv}(p_a[n],\hat{p}_j[n],p^{lo}_a[n],\hat{p}^{lo}_j[n]) \geq r[n],~\forall n\in\N\\
&\hspace{-15mm}\eqref{ph3_cond4}~\&~\eqref{ph3_cond6}
\end{align}
\end{subequations}
where $\alpha^{lo}_k[n], w^{lo}_k[n], s^{lo}_k[n], \forall k, n$ are the given local points. Note that too many approximations are used to obtain the convex problem of $(\widetilde{\P3.1})$.

\section{Proof of Lemma \ref{lemma1:rank1refomulaiton}}\label{Appendix B}
\renewcommand{\A}{\ensuremath{\mathbf{A}}}
\renewcommand{\R}{\ensuremath{\mathbf{R}}}
\renewcommand{\S}{\ensuremath{\mathbf{S}}}
\renewcommand{\B}{\ensuremath{\mathbf{B}}}
\renewcommand{\v}{\ensuremath{\mathbf{v}}}
\renewcommand{\T}{\ensuremath{\mathbf{T}}}

Given $\A$ be a complex square PSD matrix, we want to prove that $\A$ is Hermitian with nonnegative real eigenvalues. To this end, we first prove that every complex matrix can be written $\A$ can be written as $\A=\R+j\S$ where $\R$ and $\S$ are two Hermitian matrices. We can commence by rewriting $\A$ as
\begin{align}
    \A = \frac{1}{2}(\A+\A^\dagger)+\frac{1}{2}(\A-\A^\dagger), 
\end{align}
The first term of RHS, i.e. $\R = \frac{1}{2}(\A+\A^\dagger)$ is Hermitian, since $\R^\dagger=\R$ while the second term $\T = \frac{1}{2}(\A-\A^\dagger)$ is skew-Hermitian, i.e., $\T^\dagger=-\T$. Notice that if $\T$ is skew-Hermitian, then $\S=-j\T$ must be Hermitian, since $\S^\dagger=(-j\T)^\dagger=-j\T=\S$. Therefore $\A=\R+j\S$ where $\R$ and $\S$ are both Hermitian. This representation is indeed unique. 
We use contradiction approach to prove the uniqueness of this decomposition.  Suppose that $\A$ can also be written as $\A=\R^\dagger+j\S^\dagger$ with $\R^\dagger\neq \R$ and $\S^\dagger\neq \S$ being two Hermitian matrices. By subtracting the two forms of representations from each other we obtain $\R-\R^\dagger = -j (\S-\S^\dagger)$. Since $\R, \R^\dagger, \S, \S^\dagger$ are assumed to be Hermitian matrices, thus, $\R-\R^\dagger$ and $\S-\S^\dagger$ must also be Hermitian. However, here we have Hermitian matrix $\R-\R^\dagger$ being equal to a skew-Hermitian matrix $-j (\S-\S^\dagger)$ which is possible only if only possible if  $\R-\R^\dagger= -j(\S-\S^\dagger)=0$ implying that $\R-\R^\dagger$ and $\S=\S^\dagger$. So, the Hermitian decomposition of a complex square matrix $\A=\R+j\S$ is unique. 
Assuming that $\A\in\mathbb{S}^+$ then by the definition $\v^\dagger \A \v \geq 0$ for an arbitrary vector $\v$. Letting $\v$ be the normalized eigenvector of matrix $\S$ with corresponding eigenvalue $\lambda_v$, then $\v^\dagger \A \v  = \v^\dagger \B \v  + j \lambda_v$. The LHS must be a non-negative real value according to the definition of PSD, thus, this is only possible if $\lambda_v=0$ resulting the fact that $\A=\B$ implying that $\A$ is Hermitian. Thus, any complex square PSD matrix is Hermitian with non-negative real eigenvalues. This completes the proof.

\bibliographystyle{IEEEtran}
\bibliography{library}

\begin{IEEEbiography}[{\includegraphics[width=1in,height=1.25in,clip,keepaspectratio]{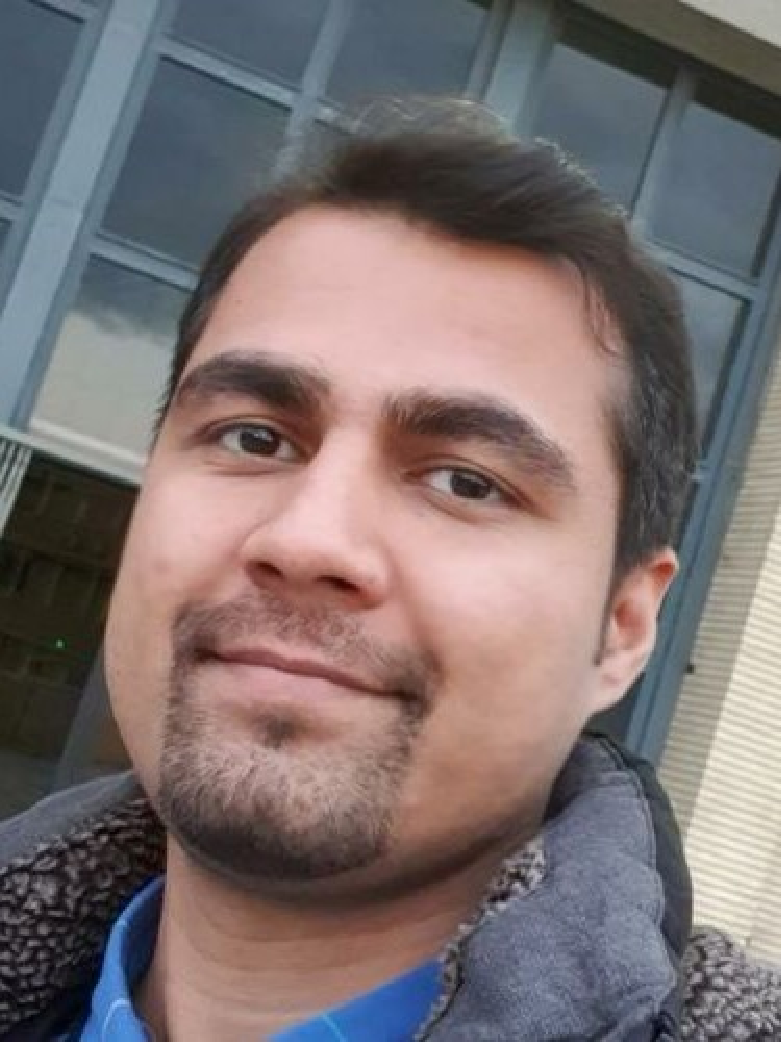}}]{Milad Tatar Mamaghani}(GS'20) was born in Tabriz, Iran, on May 12, 1994. He earned dual B.Sc. (Hons) degrees in electrical engineering fields - Telecommunications and Control - from the Amirkabir University of Technology, Tehran, Iran, in 2016 and 2018, respectively. He is currently pursuing the Ph.D. degree with the Department of Electrical and Computer Systems Engineering, Monash University, Melbourne, Australia. He is the author of several papers published in prestigious journals, and has served as a volunteer reviewer of various reputable publication venues such as TWC, TIFS, TVT, TCOM, TCCN, TMC, ISJ, Access, WCL, etc. His research interests mainly focus on beyond 5G wireless communications and networking, physical-layer security, UAV communications, optimization, and artificial intelligence. He is a member of the IEEE Communications Society and the IEEE Signal Processing Society.
\end{IEEEbiography}

 \vskip 0pt plus -1fil

\begin{IEEEbiography}[{\includegraphics[width=1in,height=1.25in,clip,keepaspectratio]{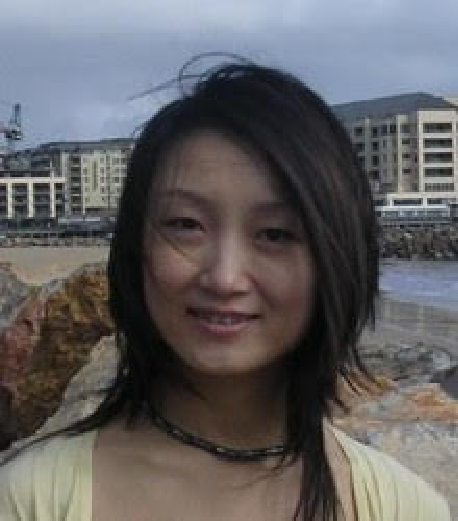}}]%
{Yi Hong}(S'00--M'05--SM'10)
is currently an Associate Professor at the Department of Electrical and Computer Systems Eng.,
Monash University, Melbourne, Australia.
She obtained her Ph.D. degree in Electrical Engineering and Telecommunications 
from the University of New South Wales (UNSW), Sydney, and received   
the {\em NICTA-ACoRN Earlier Career Researcher Award} at the 2007 {\em Australian Communication
Theory Workshop}, Adelaide, Australia. She served on the Australian Research Council College of Experts (2018-2020). Yi Hong is currently an Associate Editor for {\em IEEE Transactions on Green Communications and Networking}, and was the Associate Editor for {\em IEEE Wireless Communications Letters} and {\em Transactions on Emerging Telecommunications Technologies (ETT)}.
She was the General Co-Chair of {\em IEEE Information Theory Workshop} 2014, Hobart; the Technical Program Committee Chair of
{\em Australian Communications Theory Workshop} 2011, Melbourne; and the Publicity Chair
at the {\em IEEE Information Theory Workshop} 2009, Sicily. She was a Technical Program Committee member for
many IEEE leading conferences. Her research interests include
communication theory, coding and information theory with applications to telecommunication engineering.
\end{IEEEbiography}

\end{document}